\documentclass{elsarticle}
\usepackage{graphicx}
\usepackage{amsmath}
\usepackage{latexsym}
\usepackage{textcomp}
\usepackage{epsfig}
\usepackage{subfigure}
\usepackage{xspace}
\usepackage{hyperref}
\usepackage{todonotes}
\usepackage{paralist} % for inparaenum
\usepackage{lineno}

\newcommand{\consistent}{\textit{consistent}}
\newcommand{\emb}{\mathcal{E}}
\newcommand{\skeleton}{\mathrm{skel}}
\newcommand{\yield}{\mathrm{yield}}
\newcommand{\pert}{\mathrm{Pert}}
\newcommand{\skeletal}{\mathrm{Skel}}

\newcommand{\pqchoosable}{{\sc FPQ-Choosable Planarity Testing}\xspace}

\newtheorem{corollary}{Corollary}
\newtheorem{theorem}{Theorem}
\newtheorem{lemma}{Lemma}
\newproof{proof}{Proof}

\journal{Journal of Computer and System Sciences}

\begin{document}
	
\begin{frontmatter}

\title{Graph Planarity Testing with\\Hierarchical Embedding Constraints}

\author[inst1]{Giuseppe Liotta}
\ead{giuseppe.liotta@unipg.it}

\author[inst2]{Ignaz Rutter}
\ead{rutter@fim.uni-passau.de}

\author[inst1]{Alessandra~Tappini\corref{cor1}}
\ead{alessandra.tappini@studenti.unipg.it}

\cortext[cor1]{Corresponding author}

\address[inst1]{Dipartimento di Ingegneria, Universit\`a degli Studi di Perugia, Italy}
\address[inst2]{Department of Computer Science and Mathematics, University of Passau, Germany}

\begin{abstract}
	Hierarchical embedding constraints define a set of allowed cyclic orders for the edges incident to the vertices of a graph. These constraints are expressed in terms of FPQ-trees. FPQ-trees are a variant of PQ-trees that includes F-nodes in addition to P- and to Q-nodes. An F-node represents a permutation that is fixed, i.e., it cannot be reversed. Let $G$ be a graph such that every vertex of $G$ is equipped with a set of FPQ-trees encoding  hierarchical embedding constraints for its incident edges. We study the problem of testing whether $G$ admits a planar embedding such that, for each vertex $v$ of $G$, the cyclic order of the edges incident to $v$ is described by at least one of the FPQ-trees associated with~$v$.
	We prove that the problem is fixed-parameter tractable for biconnected graphs, where the parameters are the treewidth of $G$ and the number of FPQ-trees associated with every vertex of $G$.  We also show that the problem is NP-complete if parameterized by the number of FPQ-trees only, and  W[1]-hard if parameterized by the treewidth only.
	Besides being interesting on its own right, the study of planarity testing with hierarchical embedding constraints can be used to address other planarity testing problems. In particular, we apply our techniques to the study of NodeTrix planarity testing of clustered graphs. We show that NodeTrix planarity testing with fixed sides is fixed-parameter tractable when parameterized by the size of the clusters and by the treewidth of the multi-graph obtained by collapsing the clusters to single vertices, provided that this graph is biconnected.
\end{abstract}


\end{frontmatter}

%\linenumbers
      	
\section{Introduction}\label{se:introduction}

The study of graph planarity testing and of its variants is at the heart of graph algorithms and of their applications in various domains (see, e.g.,~\cite{p-pte-13}). Among the most studied variants we recall, for example, upward planarity testing, rectilinear planarity testing, clustered planarity testing, and HV-planarity testing (see, e.g.,~\cite{br-npcpcep-16,cd-cp-socg05,dlp-hvpac-19,fce-pcg-95,gt-upt-95,gt-ccuprp-01}). This paper studies a problem of graph planarity testing subject to embedding constraints.

In its more general terms, graph planarity with embedding constraints addresses the problem of testing whether a graph $G$ admits a planar embedding where the cyclic order of the edges incident to (some of) its vertices is totally or partially fixed. For example, Angelini et al.~\cite{adfjkpr-tppeg-15} and Jel\'inek et al.~\cite{jr-kttppeg-13} study the case when the planar embedding of a subgraph $H$ of $G$ is given as part of the input. Angelini et al.~\cite{adfjkpr-tppeg-15} present a linear-time solution to the problem of testing whether $G$ admits a planar embedding that extends the given embedding of $H$.  Jel\'inek et al.~\cite{jr-kttppeg-13} show that if the planarity test fails, an obstruction taken from a collection of minimal non-planar instances can be produced in polynomial time. A different planarity testing problem with embedding constraints is studied by Dornheim~\cite{d-pgtc-02}, who considers the case that $G$ is given with a distinguished set of cycles and it is specified, for each cycle, that certain edges must lie inside or outside the cycle. He proves NP-completeness in general and describes a polynomial-time solution when the graph is biconnected and any two cycles  share at most one vertex.  Da~Lozzo and Rutter~\cite{dr-aafcp-18} give an approximation algorithm for a restricted version of the problem. 

The research in this paper is inspired by a seminal work of Gutwenger~et~al. \cite{gkm-ptoei-08} who study the graph planarity testing problem subject to hierarchical embedding constraints.
Hierarchical embedding constraints specify for each vertex $v$ of $G$ which cyclic orders of the edges incident to $v$ are admissible in a constrained planar embedding of $G$. The term ``hierarchical'' reflects the  fact that these constraints describe ordering relationships both between sets of edges incident to a same vertex and, recursively, between edges within a same set. For example, Fig.~\ref{fi:embedding-constraints} shows a vertex~$v$, its incident edges, and a set of hierarchical embedding constraints on these edges. The edges in the figure are partitioned into four sets, denoted as $E_1$, $E_2$, $E_3$, and $E_4$; the embedding constraints allow only two distinct clockwise cyclic orders for these edge-sets, namely either $E_1E_2E_3E_4$ (Fig.~\ref{fi:embedding-constraints-a}) or $E_1E_3E_2E_4$ (Fig.~\ref{fi:embedding-constraints-b}). Within each set, the hierarchical embedding constraints of Fig.~\ref{fi:embedding-constraints} allow  the edges of  $E_1$, $E_2$, and $E_3$ to be arbitrarily permuted with one another, while the edges of $E_4$ are partitioned into three subsets $E_4'$, $E_4''$, and $E_4'''$ such that $E_4''$ must always appear between $E_4'$ and $E_4'''$ in the clockwise order around $v$. Also, the edges of $E_4'$ can be arbitrarily permuted, while the edges of $E_4''$ and the edges of $E_4'''$ have only two possible orders that are the reverse of one another.
\begin{figure}[tb]
	\centering
	\subfigure[]{\includegraphics[width=.85\textwidth,page=7]{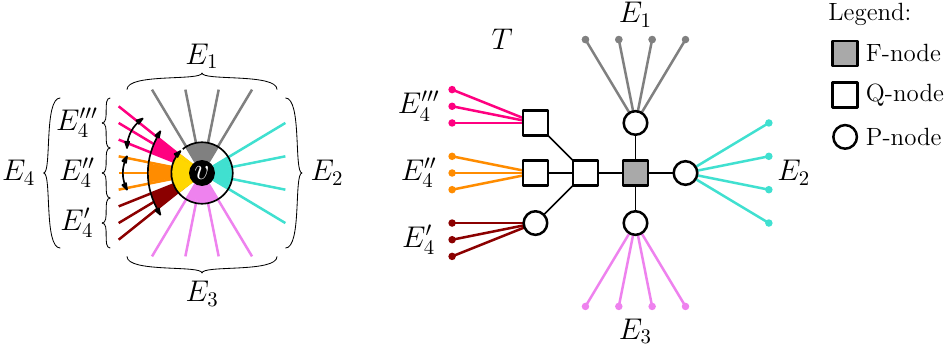}\label{fi:embedding-constraints-a}}
	\hfil
	\subfigure[]{\includegraphics[width=.85\textwidth,page=8]{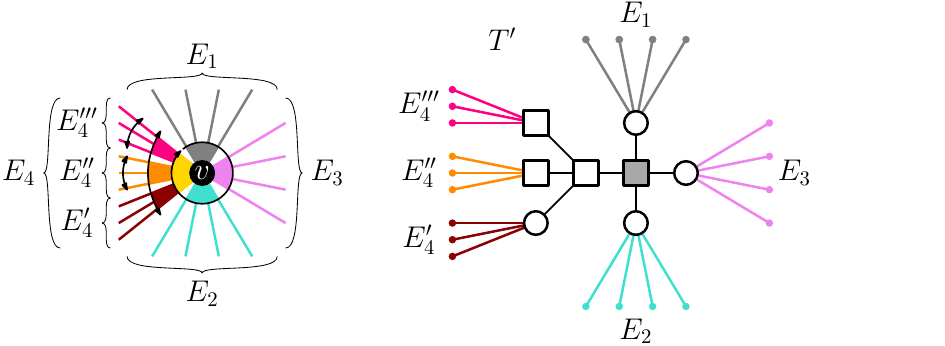}\label{fi:embedding-constraints-b}}
	\caption{(a)-(b) Two examples of a vertex $v$ with hierarchical embedding constraints and the corresponding FPQ-trees.}
	\label{fi:embedding-constraints}
\end{figure}

Hierarchical embedding constraints can be conveniently encoded by using FPQ-trees,  a variant of PQ-trees that includes F-nodes in addition to P-nodes and to Q-nodes. An F-node encodes a permutation that cannot be reversed. For example, the hierarchical embedding constraints of Fig.~\ref{fi:embedding-constraints} can be represented by two FPQ-trees denoted as $T$ and $T'$ in Fig.~\ref{fi:embedding-constraints-a} and ~\ref{fi:embedding-constraints-b}, respectively. The leaves of $T$ and $T'$ are the elements of $E_1,E_2,E_3,E'_4,E_4''$, and $E'''_4$. In the figure, F-nodes are depicted as shaded boxes, Q-nodes as white boxes, and P-nodes as circles.
The F-node of the FPQ-tree $T$ in Fig.~\ref{fi:embedding-constraints-a} enforces the cyclic order $E_1E_2E_3E_4$ around $v$, while the F-node of the FPQ-tree $T'$ in Fig.~\ref{fi:embedding-constraints-b} enforces the cyclic order $E_1E_3E_2E_4$. Both in $T$ and in $T'$, the Q-node that is adjacent to the F-node enforces $E_4''$ to appear between $E_4'$ and $E_4'''$ in clockwise order around $v$. The constraints by which the edges of $E_1$, $E_2$, $E_3$, and $E'_4$ can be arbitrarily permuted around $v$ are encoded by P-nodes in $T$ and in $T'$.

Gutwenger et al.~\cite{gkm-ptoei-08} study the planarity testing problem with hierarchical embedding constraints by allowing {\em at most one} FPQ-tree per vertex. In this paper we generalize their study and allow {\em more than one} FPQ-tree associated with each vertex.
Besides being interesting on its own right, this generalization can be used to model and study other graph planarity testing problems. As a proof of concept, we apply our results to the study of NodeTrix planarity testing of clustered graphs.

\begin{figure}[tb]
	\centering
	\subfigure[]{\includegraphics[width=.47\textwidth,page=11]{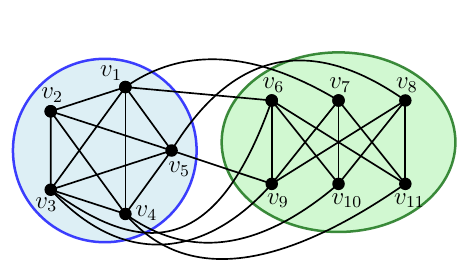}\label{fi:intro-a}}
	\hfill
	\subfigure[]{\includegraphics[width=.52\textwidth,page=12]{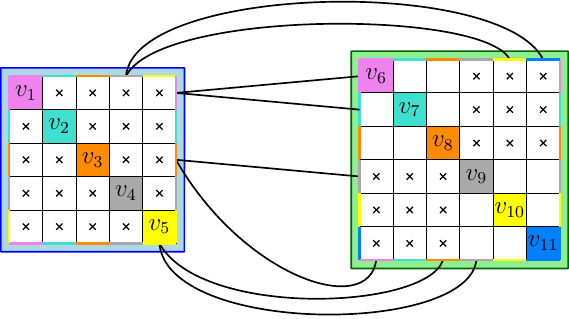}\label{fi:intro-c}}
	\caption{(a) A non-planar flat clustered graph $G$. Clusters are highlighted in blue and green. (b) A planar NodeTrix representation of $G$.}
	\label{fi:intro}
\end{figure}

Before listing our results, we recall here that NodeTrix representations have been introduced to visually explore flat clustered graphs by Henry et al.~\cite{hfm-dhvsn-07} in one of the most cited papers of the InfoVis conference~\cite{citevis}. See also~\cite{hfm-dhvsn-07,bbdlpp-valg-11,ddfp-cnrcg-jgaa-17,dlpt-ntptsc-19}.
A flat clustered graph $G$ is a graph whose vertex set is partitioned into subsets called clusters. A NodeTrix representation of $G$ represents its clusters as adjacency matrices, while the edges connecting different matrices are represented as simple curves (see for example Figure~\ref{fi:intro}). The NodeTrix planarity testing problem asks whether $G$ admits a NodeTrix representation without edge crossings. The question can be asked both in the ``fixed sides'' scenario and in the ``free sides'' scenario. The fixed sides scenario specifies, for each edge $e$ connecting two matrices $M$ and $M'$, the sides (Top, Bottom, Left, Right) of $M$ and $M'$ to which $e$ must be incident; in the free sides scenario the testing algorithm can choose the sides to which $e$ is incident. NodeTrix planarity testing is known to be NP-complete in both scenarios~\cite{ddfp-cnrcg-jgaa-17,dlpt-ntptsc-19,bdlg-ckmepd-19}.
Our main results are the following.

\begin{itemize}
	\item We show that \pqchoosable is NP-complete even if the number of FPQ-trees associated with each vertex is bounded by a constant, and it remains NP-complete even if the FPQ-trees only contain P-nodes. This contrasts with the result of Gutwenger et al.~\cite{gkm-ptoei-08} who prove that \pqchoosable can be solved in linear time when each vertex is equipped with at most one FPQ-tree.
	We also prove that \pqchoosable is W[1]-hard parameterized by treewidth, and that it remains W[1]-hard even when the FPQ-trees only contain P-nodes.
%	\item We show that \pqchoosable is NP-complete even if the number of FPQ-trees associated with each vertex is bounded by a constant. This contrasts with the result of Gutwenger et al.~\cite{gkm-ptoei-08} who prove that \pqchoosable can be solved in linear time when each vertex is equipped with at most one FPQ-tree.
%	We also prove that \pqchoosable remains NP-complete even if the FPQ-trees associated with the vertices only contain P-nodes.
	
	\item The above results imply that \pqchoosable is not fixed-parameter tractable if parameterized by treewidth only or by the number of FPQ-trees per vertex only. For a contrast, we show that \pqchoosable becomes fixed-parameter tractable for biconnected graphs when parameterized by both the treewidth and the number of FPQ-trees associated with every vertex.
	
	\item We show that there is a strict interplay between the \pqchoosable problem and the problem of testing whether a flat clustered graph $G$ is NodeTrix planar. Indeed, we prove that NodeTrix planarity testing with fixed sides is fixed-parameter tractable when parameterized by the size of the clusters of $G$ and by the treewidth of the multi-graph obtained by collapsing the clusters of $G$ to single vertices, provided that this graph is biconnected. If we consider the vertex degree of $G$ as an additional parameter, the fixed-parameter tractability immediately extends to NodeTrix planarity testing with free sides.
\end{itemize}

From a technical point of view,  our algorithmic approach is based on a combined usage of different data structures, namely SPQR-trees~\cite{dt-olpt-96}, FPQ-trees, and sphere-cut decomposition trees~\cite{dpbf-eeapg-10,gt-obdpg-08,st-crr-94}.
It may be worth recalling that a polynomial-time solution for NodeTrix planarity testing with fixed sides was known only when the size of each cluster is bounded by a constant and the treewidth of the graph obtained by collapsing the clusters to single vertices is two~\cite{dlpt-ntptsc-19}.

%We finally remark that, besides NodeTrix planarity testing, our algorithmic approach can be used to solve any constrained planarity testing problem that can be modeled by equipping each vertex of a graph with a set of FPQ-trees. As an example, we give a corollary about a polynomial time algorithm for flat clustered planarity testing when the vertex degree and the branchwidth of the graph of the clusters are bounded by a constant.

The rest of the paper is organized as follows. Section~\ref{se:preliminaries} reports preliminary definitions. Section~\ref{se:problem} introduces the \pqchoosable problem,
%Section~\ref{se:pqchoosable-testing-hard} shows the NP-completeness of \pqchoosable,
Section~\ref{se:pqchoosable-testing-hard} studies its computational complexity,
in Section~\ref{se:pqchoosable-testing-fpt} we describe a fixed-parameter tractability approach for \pqchoosable, and in Section~\ref{se:nodetrix} we analyze the interplay between \pqchoosable and NodeTrix Planarity testing. Concluding remarks and open problems are given in Section~\ref{se:open-problems}.

\section{Preliminaries}\label{se:preliminaries}
	
	We assume familiarity with  graph theory and algorithms, and we only briefly recall some of the basic concepts that will be used extensively in the rest of the paper (see also~\cite{arumugam2016handbook,dett-gd-99}).
	
	A \emph{PQ-tree} is a tree-based data structure that represents a family of permutations on a set of elements~\cite{bl-tcopiggpupa-76}. In a PQ-tree, each element is represented by one of the leaf nodes, and each non-leaf node is a \emph{P-node} or a \emph{Q-node}. The children of a P-node can be permuted arbitrarily, while the order of the children of a Q-node is fixed up to reversal.
%	Given a graph $G$ together with a fixed combinatorial embedding, we can associate with each vertex $v$ a PQ-tree $T_v$ whose leaves represent the edges incident to $v$, so that $T_v$ represents a set of cyclic orders of the edges around $v$.
	Given a graph $G$ together with a fixed combinatorial embedding, we can associate with each vertex $v$ a PQ-tree $T_v$ whose leaves represent the edges incident to $v$. Tree $T_v$ encodes a set of permutations for its leaves: Each of these permutations is in a bijection with a cyclic order of the edges incident to $v$. If there is a permutation $\pi_v$ of the leaves of $T_v$ that is in a bijection with a cyclic order $\sigma_v$ of the edges incident to $v$, we say that $T_v$ \emph{represents} $\sigma_v$, or equivalently that $\sigma_v$ \emph{is represented by} $T_v$.
	
	An \emph{FPQ-tree} is a PQ-tree where, for some of the Q-nodes, the reversal of the permutation described by their children is not allowed. To distinguish these Q-nodes from the regular Q-nodes, we call them \emph{F-nodes}.
%	More precisely, an FPQ-tree has three types of nodes: F-nodes, P-nodes, and Q-nodes. The cyclic order of the children of an F-node is totally fixed, while the children of a P-node can be permuted arbitrarily, and the cyclic order of the children of a Q-node is fixed up to reversal.
	It may be worth recalling that Gutwenger et al.~\cite{gkm-ptoei-08} call this data structure ``embedding constraint'', and that their ``gc-nodes'' correspond to P-nodes, ``mc-nodes'' to Q-nodes, and ``oc-nodes'' to F-nodes.
	
	Let $G$ be a biconnected planar (multi-)graph. An \emph{SPQR-decomposition} of $G$ describes the structure of $G$ in terms of its triconnected components by means of a tree called the \emph{SPQR-decomposition tree}, and denoted as $\mathcal{T}$ (see, e.g.,~\cite{dt-olpt-96,dett-gd-99}). Tree  $\mathcal{T}$ can be computed in linear time and it has three types of internal nodes that correspond to different arrangements of the components of $G$. If the components are arranged in a cycle, they correspond to an \emph{S-node} of $\mathcal{T}$; if they share two vertices and are arranged in parallel, they correspond to a \emph{P-node} of $\mathcal{T}$; if they are arranged in a triconnected graph, they correspond to an \emph{R-node} of $\mathcal{T}$.
	The leaves of $\mathcal{T}$ are \emph{Q-nodes}, and each of them corresponds to an edge of $G$.
	To simplify the description and without loss of generality, we shall assume that every S-node of $\mathcal{T}$ has exactly two children.
	For each node $\mu$ of $\mathcal{T}$, the \emph{skeleton} of $\mu$ is an auxiliary graph that represents the arrangement of the triconnected components of $G$ corresponding to $\mu$, and it is denoted by $\skeleton(\mu)$. Each edge of $\skeleton(\mu)$ is called a \emph{virtual edge}, and the end-points of a (possibly virtual) edge are called \emph{poles}. Every virtual edge corresponds to a subgraph of $G$ called the \emph{pertinent graph}, that is denoted by $G_{\mu}$.
	Tree $\mathcal{T}$ encodes all possible planar combinatorial embeddings of $G$. These embeddings are determined by P- and R-nodes, since the skeletons of S- and Q-nodes have a unique embedding. Indeed, the skeleton of a P-node consists of parallel edges that can be arbitrarily permuted, while the skeleton of an R-node is triconnected, and hence it has a unique embedding up to a flip.
	Figure~\ref{fi:ex-a} shows a biconnected planar multi-graph $G$ and Figure~\ref{fi:ex-b} illustrates an SPQR-decomposition tree of $G$.
	
	Note that the planar combinatorial embeddings that are given by the SPQR-decomposition tree of a biconnected graph $G$ give constraints on the cyclic order of edges around each vertex of $G$. 
	These constraints can be encoded by associating a PQ-tree to each vertex $v$ of $G$, called the \emph{embedding tree} of $v$ and denoted as $T_v^\epsilon$ (see, e.g.,~\cite{br-spoace-16}). For example, Figure~\ref{fi:ex-c} shows the embedding tree $T_{v_2}^\epsilon$ of the vertex $v_2$ in Figure~\ref{fi:ex-a}. Note that edges $f$ and $h$ ($i$ and $j$, resp.) belong to an R-node (a P-node, resp.) in the SPQR-decomposition tree of $G$, hence the corresponding leaves are connected to a Q-node (a P-node, resp.) in~$T_{v_2}^\epsilon$.

	\begin{figure}[tb]
		\centering
		\subfigure[]{\includegraphics[width=.21\textwidth,page=1]{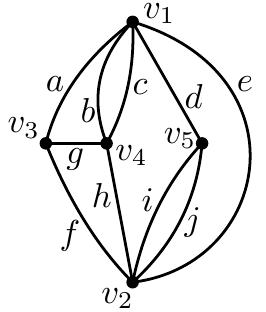}\label{fi:ex-a}}
		\hfil
		\subfigure[]{\includegraphics[width=.63\textwidth,page=6]{ex}\label{fi:ex-b}}
		\hfil
		\subfigure[]{\includegraphics[width=.13\textwidth,page=7]{ex}\label{fi:ex-c}}
		\caption{(a) A biconnected planar graph $G$. (b) An SPQR-decomposition tree of $G$. The skeletons of S-, P-, and R- nodes are inside gray boxes, while Q-nodes are depicted as letters. (c) The embedding tree of $v_2$.}
		\label{fi:ex}
	\end{figure}

\section{The FPQ-choosable Planarity Testing Problem}\label{se:problem}

Let $G=(V,E)$ be a (multi-)graph, let $v\in V$, and let $T_v$ be an FPQ-tree whose leaf set is $E(v)$, i.e., the set of the edges incident to $v$. 
We define $\consistent(T_v)$ as the set of cyclic orders of the edges incident to $v$ in $\emb$ that are represented by the FPQ-tree $T_v$.

An \emph{FPQ-choosable graph} is a pair $(G,D)$ where $G=(V,E)$ is a (multi-) graph, and $D$ is a mapping that associates each vertex $v\in V$ with a set $D(v)$ of FPQ-trees whose leaf set is $E(v)$.
Given a planar embedding $\emb$ of $G$, we denote by $\emb(v)$ the cyclic order of edges incident to $v$ in $\emb$.
%In the following, we omit $D$ from the definition of FPQ-choosable graph when the sets of FPQ-trees associated with the vertices of $G$ are clear from the context.
%
An \emph{assignment} $A$ is a function that assigns to each vertex $v\in V$ an FPQ-tree in $D(v)$. We say that $A$ is \emph{compatible with $G$} if there exists a planar embedding $\emb$ of $G$ such that $\emb(v) \in \consistent(A(v))$ for all $v \in V$. In this case, we also say that $\emb$ is \emph{consistent with $A$}.

An FPQ-choosable graph $(G,D)$ is \emph{FPQ-choosable planar} if there exists an assignment of FPQ-trees that is compatible with $G$.
Figure~\ref{fi:compatible-c} shows an FPQ-choosable planar graph $G$, whose vertices are equipped with the following sets of FPQ-trees: $D(u_1)=\{T_\alpha\}$, $D(u_2)=\{T_\beta,T_\gamma\}$, $D(u_3)=\{T_\delta\}$, and $D(u_4)=\{T_\varepsilon\}$.
There are two possible assignments that differ from one another for the chosen FPQ-tree in the set $D(u_2)$. As illustrated in Figures~\ref{fi:compatible-a} and~\ref{fi:compatible-b}, the first assignment is compatible with $G$, while there is no planar embedding that is consistent with the second assignment.
%Figure~\ref{fi:compatible-a} shows an embedding consistent with an assignment that is compatible with $G$; in Figure~\ref{fi:compatible-b}, there is no planar embedding that is consistent with the shown assignment.

\begin{figure}[tb]
	\centering
	\subfigure[]{\includegraphics[width=.45\textwidth,page=3]{compatible}\label{fi:compatible-c}}\\
	\subfigure[]{\includegraphics[width=.45\textwidth,page=1]{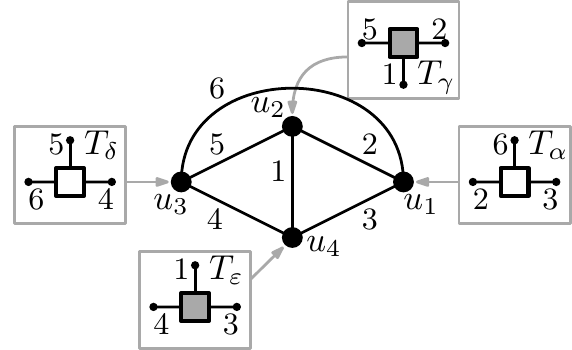}\label{fi:compatible-a}}
	\hfil
	\subfigure[]{\includegraphics[width=.45\textwidth,page=2]{compatible}\label{fi:compatible-b}}
	\caption{(a) An FPQ-choosable planar graph $(G,D)$. (b) A planar embedding of $G$ that is consistent with assignment $\{A(u_1)=T_\alpha, A(u_2)=T_\gamma, A(u_3)= T_\delta, A(u_4)=T_\varepsilon\}$; the assignment is compatible with $G$. (c) A non-planar embedding of $G$ that is consistent with assignment $\{A(u_1)=T_\alpha, A(u_2)=T_\beta, A(u_3)= T_\delta, A(u_4)=T_\varepsilon\}$; there is no planar embedding that is consistent with $A$.}
	\label{fi:compatible}
\end{figure}
The \pqchoosable problem receives as input an FPQ-choosable graph $(G,D)$ and it asks whether $(G,D)$ is FPQ-choosable planar, i.e., it asks whether there exists an assignment that is compatible with $G$. In the rest of the paper we are going to assume that $G$ is a biconnected (multi-) graph. Clearly $G$ must be planar or else the problem becomes trivial.
Also, any assignment that is compatible with $G$ must define a planar embedding of $G$ among those described by an SPQR-decomposition tree of $G$.

Therefore, a preliminary step for an algorithm that tests whether $(G,D)$ is FPQ-choosable planar is to intersect each FPQ-tree $T_v \in D(v)$ with the embedding tree $T_v^\epsilon$ of $v$, so that the cyclic order of the edges incident to $v$ satisfies both the constraints given by $T_v$ and the ones given by $T_v^\epsilon$. (See, e.g.,~\cite{br-spoace-16} for details about the operation of intersection between two PQ-trees, whose extension to the case of FPQ-trees is straightforward since F-nodes are just a more constrained version of Q-nodes). Therefore, from now on we shall assume that the FPQ-trees of $D$ have been intersected with the corresponding embedding trees and, for ease of notation, we shall still denote with $D(v)$ the set of FPQ-trees associated with $v$ and resulting from the intersection. We also remove the null-tree, which represents the empty set of permutations, from the sets $D(v)$.  Clearly, a necessary condition for the FPQ-choosable planarity of $(G,D)$ is that $D(v)$ is not the empty set for every $v \in G$.

	\section{Complexity of FPQ-choosable Planarity Testing}\label{se:pqchoosable-testing-hard}
	
%	Gutwenger et al.~\cite{gkm-ptoei-08} show that the \pqchoosable problem can be solved in $O(n)$ time for an FPQ-choosable graph $(G,D)$ such that $|D(v)| \leq 1$ for every vertex $v$ of $G$. As shown by the following lemma, in its generality the problem is NP-complete even if $(G,D)$ is such that $|D(v)|$ is bounded by a constant for every vertex $v$, and even if $D$ consists of FPQ-trees having only P-nodes.

	As we are going to show, \pqchoosable is fixed-parameter tractable when parameterized by treewidth and number of FPQ-trees per vertex. One may wonder whether the problem remains fixed-parameter tractable if parameterized by the treewidth only or by the number of FPQ-trees per vertex only. The following two theorems answer this question in the negative.
	
 	\begin{theorem}\label{th:pqchoosable-npcomplete}
 		\pqchoosable with a bounded number of FPQ-trees per vertex is NP-complete. It remains NP-complete even when the FPQ-trees have only P-nodes.
 	\end{theorem}	
	\begin{proof}
 	We denote with $n$ the number of vertices of the input graph and we assume that for each vertex $v$ of the input, $|D(v)| \in O(n)$. We generate all possible assignments by performing $O(n \log n)$ non-deterministic guess operations and, for each assignment, we decide whether it is compatible with the input graph by applying the linear-time algorithm of Gutwenger et al.~\cite{gkm-ptoei-08}. It follows that \pqchoosable is in NP. 
 
 	In order to show that \pqchoosable is NP-hard, we use a reduction from the problem of deciding whether a triconnected cubic graph admits a $3$-edge-coloring. The $3$-edge-coloring problem for a cubic graph asks whether it is possible to assign a color in the set $\{red, green, blue\}$ to each edge of the graph so that no two edges of the same color share a vertex.
 	The problem is known to be NP-complete for triconnected cubic non-planar graphs~\cite{h-npcec-81}.
 	Note that a triconnected cubic graph admits a $3$-edge-coloring if and only if it admits a $3$-edge-coloring for any choice of rotation system and outer face, hence we perform the construction starting from a triconnected cubic graph with an arbitrary choice of rotation system and outer face, which makes it possible to talk about edge crossings in the graph.
 	For any given triconnected cubic graph $G$ we construct an FPQ-choosable graph $(G',D')$ with $|D(v')\le 6|$ for each vertex $v'$ of $G'$, that is FPQ-choosable planar if and only if $G$ has a $3$-edge-coloring. Since every vertex of $(G',D')$ is equipped with at most six FPQ-trees, the statement will follow. See Figure~\ref{fi:cubic} for an example.
 	
 	\begin{figure}[tbp]
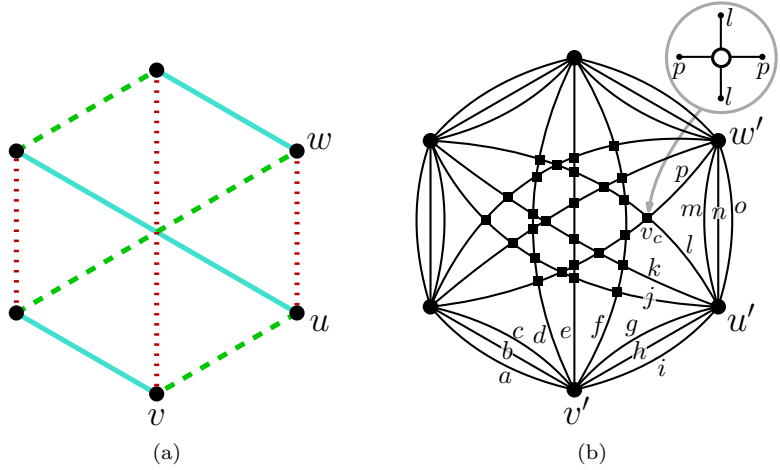
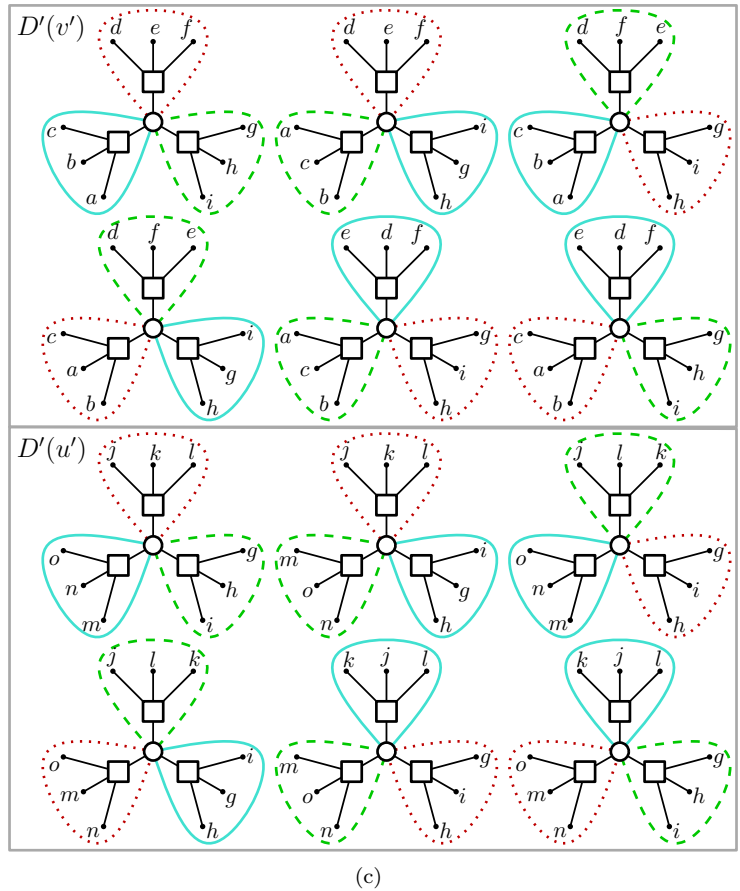

 		\centering
 		\subfigure[]{\includegraphics[width=.36\textwidth,page=5]{3edgecoloring}\label{fi:cubic-a}}
 		\hfil
 		\subfigure[]{\includegraphics[width=.4\textwidth,page=6]{3edgecoloring}\label{fi:cubic-b}}
 		\hfil
 		\subfigure[]{\includegraphics[width=.8\textwidth,page=7]{3edgecoloring}\label{fi:cubic-c}}
 		\caption{(a) A triconnected cubic non-planar graph $G$ with a proper $3$-edge-coloring. (b) The corresponding FPQ-choosable graph $(G',D')$; the dummy vertices are black squares, and the FPQ-tree associated with $v_c$ is inside a gray circle. $D'(v')$ and $D'(u')$ are shown in (c).}
 		\label{fi:cubic}
 	\end{figure}
	 
	The construction that maps any triconnected cubic graph $G$ into an FPQ-choosable graph $(G',D')$ is as follows. Each vertex $v$ of $G$ is associated with a vertex $v'$ in $G'$, and each edge $e=(u,v)$ of $G$ is associated in $G'$ with three parallel edges $e^1$, $e^2$, and $e^3$,
%	paths of arbitrary length $\pi_e^1$, $\pi_e^2$, and $\pi_e^3$,
	whose end-vertices are $u'$ and $v'$.  See for example Figure~\ref{fi:3edgecoloring-a} and~\ref{fi:3edgecoloring-b}.
%	(by parallel paths we mean that no two such paths share a vertex except $u'$ and $v'$).
	Each crossing $c$ of $G'$ is replaced with a dummy vertex~$v_c$. Note that every vertex of $G'$ has either degree $4$ or $9$, since we can assume that each crossing is the intersection of exactly two edges (otherwise a small perturbation can be applied).
%	Each vertex $u'$ of $G'$ having degree $2$ is equipped with one FPQ-tree $T_{u'}$ consisting of an F-node whose leaves represent the two edges incident to $u'$.
	Each vertex $v_c$ of $G'$ having degree $4$ is equipped with one FPQ-tree $T_{v_c}$ consisting of a P-node whose leaves represent the four edges incident to $v_c$.
	Each vertex $v'$ of $G'$ having degree $9$ is equipped with a set $D'(v')$ of FPQ-trees. Each FPQ-tree in $D'(v')$ consists of a P-node $\rho$ connected to three Q-nodes $\chi_{e_1}$, $\chi_{e_2}$, and $\chi_{e_3}$, which have three leaves each, denoted as $p_{e_i}^1$, $p_{e_i}^2$, $p_{e_i}^3$. See for example Figure~\ref{fi:3edgecoloring-c}, that shows an FPQ-tree of the vertex $v'$ in  Figure~\ref{fi:3edgecoloring-b}. 
	
	Observe that every FPQ-tree in $D'(v')$ can be defined as the union of three trees $T_1$, $T_2$, and $T_3$, such that each $T_i$ consists of node $\rho$, node $\chi_{e_i}$, and the three leaves of $\chi_{e_i}$ ($1 \le i \le 3$). For example, $T_1$, $T_2$, and $T_3$ are highlighted in Figure~\ref{fi:3edgecoloring-c}.  Consider a Q-node $\chi_{e_i}$ and the cyclic order $\sigma_i$ of its incident edges in $T_i$. If the leaves of $T_i$ appear as $p_{e_i}^1$, $p_{e_i}^2$, $p_{e_i}^3$ in $\sigma_i$, we say that $T_i$ has a \emph{red configuration}; if they appear as $p_{e_i}^1$, $p_{e_i}^3$, $p_{e_i}^2$, we say that $T_i$ has a \emph{green configuration}; if they appear as $p_{e_i}^2$, $p_{e_i}^1$, $p_{e_i}^3$, we say that $T_i$ has a \emph{blue configuration}. For example, in Figure~\ref{fi:3edgecoloring-c} $T_1$ has a red configuration, $T_2$ has a green configuration, and $T_3$ has a blue configuration.
	 
	Let $e_1$, $e_2$, and $e_3$ be the three edges incident to a vertex $v$ in the triconnected cubic graph $G$ and let $v'$ be its corresponding vertex in $(G',D')$. For each $3$-edge-coloring of $G$, there is a bijection between an FPQ-tree $T_{v'}$ in $D'(v')$ and the colors of the three edges incident to $v$.  
	Namely, for a $3$-edge-coloring of $G$ where $e_i$ is red, we impose a red configuration to $T_i$ in $T_{v'}$; if $e_i$ is green, we impose a green configuration to $T_i$; if $e_i$ is blue, we impose a blue configuration to $T_i$. We say that \emph{$T_i$ matches the color of $e_i$} and that \emph{$T_{v'}$ matches the color of the edges incident to $v$}. 
	For example, the FPQ-tree of Figure~\ref{fi:3edgecoloring-c} matches the color of the edges incident to $v$ in Figure~\ref{fi:3edgecoloring-a}, because $T_1$ matches the color of $e_1$, $T_2$ matches the color of $e_2$ and $T_3$ matches the color of $e_3$.  Since there are six possible permutations of the three colors around $v$ in $G$, we have that $|D'(v')|=6$ in $(G',D')$. 
	 
	We now prove that if $G$ admits a $3$-edge-coloring, $(G',D')$ is FPQ-choosable planar. Let $v$ be any vertex of $G$ with incident edges $e_1$, $e_2$, $e_3$, and let $v'$ be the vertex that corresponds to $v$ in $(G',D')$. We define an assignment $A$ for $(G',D')$ where $A(v')$ is the FPQ-tree $T_{v'} \in D'(v')$ that matches the color of the edges incident to $v$.
%	For every vertex $u'$ of $(G',D')$ of degree $2$, $A(u')$ is the only FPQ-tree associated with $u'$,
	For every vertex $v_c$ of $(G',D')$ of degree $4$, $A(v_c)$ is the only FPQ-tree associated with $v_c$, hence the cyclic order of the edges around $v_c$ is totally free.
	We show that there exists a planar embedding of $G'$ that is consistent with $A$.
%	Each edge of $G'$ has an end-point $u'$ of degree $2$ and an end-point $v'$ of degree $9$. Since $T_{u'}$ consists of an F-node with two incident edges representing the edges incident to $u'$ in $G'$, the cyclic order of the edges around $u'$ is fixed. Every vertex of degree $2$ belongs to one of three parallel paths connecting the same pair $u'$ and $v'$ of two vertices of  degree $9$ in $G'$.
	Since $T_{u'}$ matches the color of the edges incident to $u$ in $G$ and $T_{v'}$ matches the color of the edges incident to $v$ in $G$, the leaves of $T_{u'}$ and the leaves of $T_{v'}$ representing the edges (possibly subdivided by dummy vertices) connecting $u'$ and $v'$ can be ordered so to avoid edge crossings. If, for example, edge $e=(u,v)$ is red in $G$, we have that $T_{u'}$ has a subtree $T'$ and $T_{v'}$ has a subtree $T''$ such that both $T'$ and $T''$ match the red color. The sets of leaves of $T'$ and $T''$ represent the same set of edges, and they appear in reverse order around $u'$ and around $v'$ in a planar embedding of $G'$. It follows that if $G$ admits a $3$-edge-coloring, $(G',D')$ is FPQ-choosable planar.
	
	\begin{figure}[tbp]
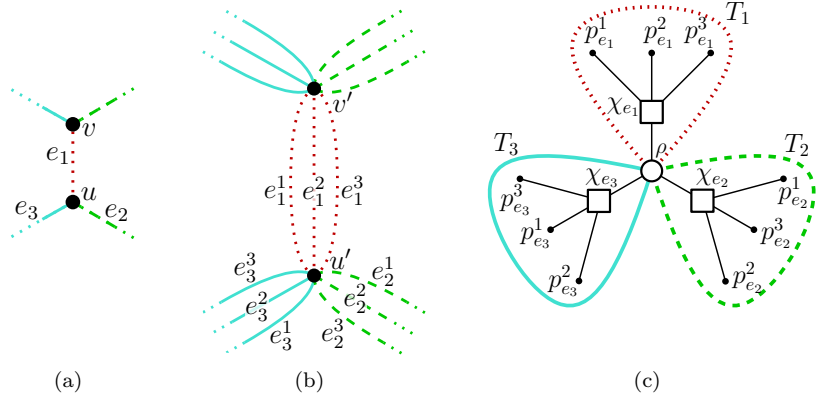

		\centering
		\subfigure[]{\includegraphics[width=.26\textwidth,page=3]{3edgecoloring}\label{fi:3edgecoloring-a}}
		\hfil
		\subfigure[]{\includegraphics[width=.25\textwidth,page=4]{3edgecoloring}\label{fi:3edgecoloring-b}}
		\hfill
		\subfigure[]{\includegraphics[width=.47\textwidth,page=2]{3edgecoloring}\label{fi:3edgecoloring-c}}
		\caption{(a) A vertex $u$ of a triconnected cubic graph $G$ and its incident edges $e_1=(u,v)$, $e_2$, and $e_3$. (b) Three parallel edges of $G'$ that are associated with edge $e_1$ of $G$. (c) An FPQ-tree $T_{u'}$ associated with vertex $u'$: $T_1$ has a red configuration, $T_2$ has a green configuration, and $T_3$ has a blue configuration.}
		\label{fi:3edgecoloring}
	\end{figure}
	  	 
	Suppose for a converse that $(G',D')$ is FPQ-choosable planar. There exists an assignment $A$ that is compatible with $G'$. Assignment $A$ defines the cyclic order of the edges incident to each vertex in a planar embedding of $G'$. Recall that for any two vertices $u'$ and $v'$ having degree $9$, they are connected by three parallel edges (possibly subdivided by dummy vertices),
%	 paths $\pi_e^1$, $\pi_e^2$, and $\pi_e^3$,
	where $e$ is the edge of $G$ in a bijection with these three edges. Since $A$ is compatible with $G'$, the two FPQ-trees $T_{u'}=A(u')$ and $T_{v'}=A(v')$ both contain two subtrees $T'$ and $T''$ such that: (i) $T'$ and $T''$ have the same set of three leaves; (ii) these three leaves represent edges of $G'$ that correspond to $e^1$, $e^2$, and $e^3$; (iii) $T'$ and $T''$ have the same red (green, blue) configuration. We color edges $e^1$, $e^2$, and $e^3$ with the red (green, blue) color depending on the color configuration of $T'$ and of $T''$. By iterating this procedure over all triplets of edges we have that around every vertex of degree~$9$ in $G'$ there are three consecutive triplets of edges such that the edges of each triplet all have the same color and no two triplets have the same color. A $3$-edge-coloring of $G$ is therefore obtained by giving every edge $e$ of $G$ the same color as the one of the corresponding triplet $e^1$, $e^2$, and $e^3$ in $G'$. It follows that if $(G',D')$ is FPQ-choosable planar then $G$ has a $3$-edge-coloring.
	
	\begin{figure}[tbp]
		\centering
		{\includegraphics[width=.45\textwidth,page=1]{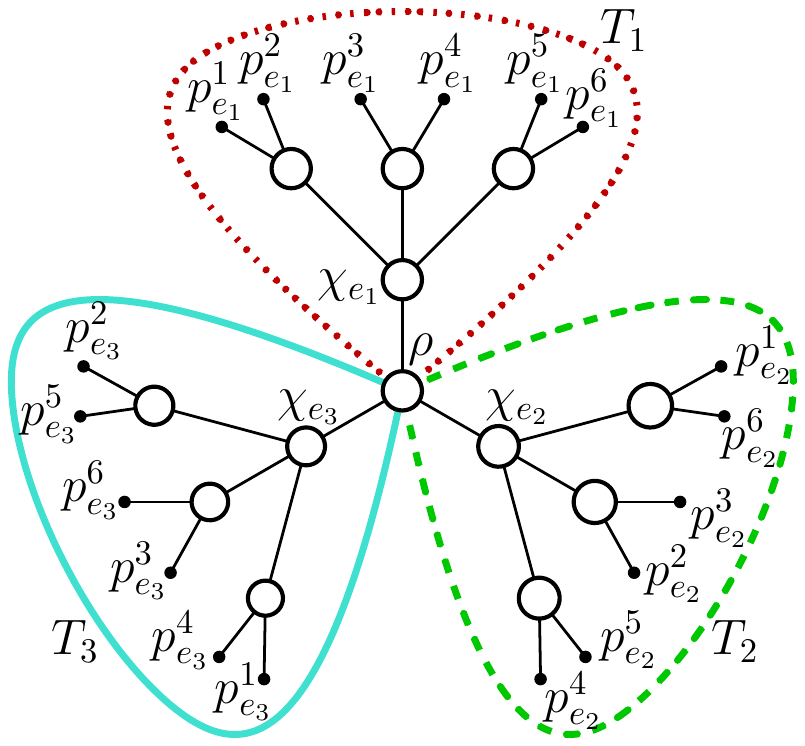}}
		\caption{An FPQ-tree with only P-nodes associated with a vertex of degree $18$ in $G'$.}
		\label{fi:3edgecoloring-pnodes}
	\end{figure}
	
	In order to prove that the problem remains NP-complete if the FPQ-trees associated with the vertices have only P-nodes, we construct an FPQ-choosable graph $(G',D')$ in a slightly different way from the one described above. In particular, each edge of $G$ is associated with \emph{six} parallel edges in $G'$, and each vertex $v'$ having degree $18$ in $G'$ is equipped with six FPQ-trees as the one in Figure~\ref{fi:3edgecoloring-pnodes}.
	In this case each FPQ-tree associated with a vertex $v'$ of $G'$ having degree $18$ is the union of three trees $T_1$, $T_2$, and $T_3$, such that each of their three pairs of leaves are connected to a P-node, which enforces each pair of leaves to appear consecutively. If $T_i$ ($1 \le i \le 3$) has a red configuration, the two leaves $(p_{e_i}^1,p_{e_i}^2)$ must be consecutive, as well as the leaves $(p_{e_i}^3,p_{e_i}^4)$, and the leaves $(p_{e_i}^5,p_{e_i}^6)$; if $T_i$ has a green configuration the two leaves $(p_{e_i}^1,p_{e_i}^6)$ must be consecutive, as well as the leaves $(p_{e_i}^3,p_{e_i}^2)$, and the leaves $(p_{e_i}^5,p_{e_i}^4)$; if $T_i$ has a blue configuration the two leaves $(p_{e_i}^1,p_{e_i}^4)$ must be consecutive, as well as the leaves $(p_{e_i}^3,p_{e_i}^6)$, and the leaves $(p_{e_i}^5,p_{e_i}^2)$. This guarantees that any two adjacent vertices $u'$ and $v'$ of $G'$ are such that if $T_{u'}$ and $T_{v'}$ match the same color, there is a cyclic order represented by $T_{u'}$ and $T_{v'}$ such that the edges incident to $u'$ and the edges incident to $v'$ do not cross. Conversely, if they match different colors these edges must cross.	\qed
	\end{proof}

We remark that Theorem~\ref{th:pqchoosable-npcomplete} also holds if the number of FPQ-trees per vertex is bounded by a constant larger than six, indeed it is possible to associate each edge of the given triconnected cubic graph $G$ with a suitable number of parallel edges and each vertex of $G$ with a suitable number of FPQ-trees.

We now prove that \pqchoosable parameterized by treewidth is W[1]-hard.

	\begin{theorem}\label{th:w1}
	\pqchoosable parameterized by treewidth is W[1]-hard. It remains W[1]-hard even when the FPQ-trees have only P-nodes.
	\end{theorem}
	\begin{proof}
		We use a parameterized reduction from the \emph{list coloring} problem, which is defined as follows. Given a graph $G=(V,E)$ and a set $L$ containing a list $L(v)$ of colors for each vertex $v\in V$, is there a proper vertex coloring with $c(v)\in L(v)$ for each $v$? We denote as $c(v)$ the color of $v$ in a proper vertex coloring. The list coloring problem parameterized by treewidth is known to be W[1]-hard even for planar graphs~\cite[Theorem~14.29]{cfklmpps-pa-15}.
		
		For any given instance $(G,L)$ of list coloring such that $G$ is a planar graph whose treewidth is at most $t$, we construct an FPQ-choosable graph $(G',D')$ that is FPQ-choosable planar if and only if $(G,L)$ is a \emph{yes} instance of list coloring. Note that $(G,L)$ is a \emph{yes} instance of list coloring if and only if it is a \emph{yes} instance for any planar embedding of $G$, hence we perform the reduction to \pqchoosable starting from any instance $(G,L)$ with an arbitrary planar embedding of $G$.
		Starting from a planar embedding of graph $G$, we construct a planarly embedded multi-graph $G'$ by replacing each edge of $G$ with bundles of edges as follows. Also refer to Figure~\ref{fi:param-reduction}.
		
		%		\begin{figure}[tb]
		%			\centering
		%			\includegraphics[width=.8\textwidth,page=4]{param}
		%			\caption{Illustration for the reduction of Lemma~\ref{le:w1}. A valid coloring is shown with white crosses.}
		%			\label{fi:param-reduction}
		%		\end{figure}
		
		\begin{figure}[tbp]
			\centering
			\subfigure[]{\includegraphics[width=.55\textwidth,page=6]{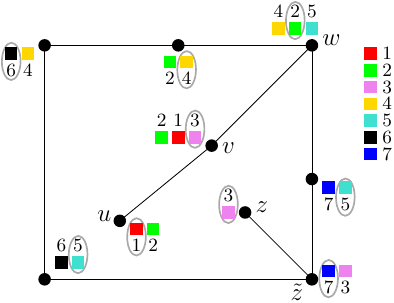}\label{fi:param-a}}
			\hfil
			\subfigure[]{\includegraphics[width=.48\textwidth,page=7]{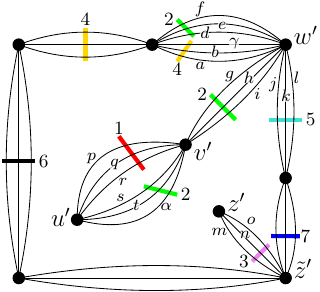}\label{fi:param-b}}
			\hfil
			\subfigure[]{\includegraphics[width=1\textwidth,page=8]{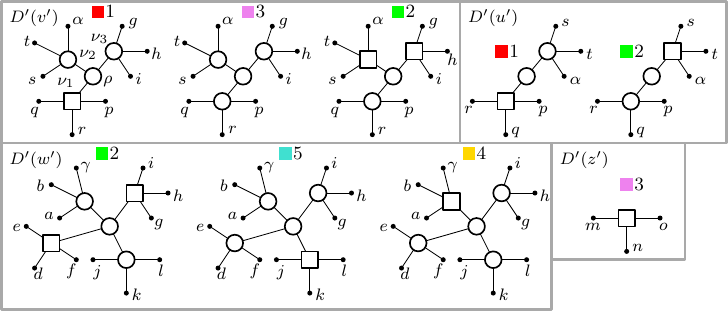}\label{fi:param-c}}
			\caption{Illustration of the reduction from list coloring to \pqchoosable. (a) An instance $(G,L)$ of list coloring. The circled colors indicate a valid coloring of $G$. (b) The corresponding FPQ-choosable graph $(G',D')$; some FPQ-trees of $D'$ are shown in (c).}
			\label{fi:param-reduction}
		\end{figure}
		
		Each vertex $v$ of $G$ becomes vertex $v'$ in $G'$, and each edge $e=(u,v)$ of $G$ is split into three parallel edges $e'_1, e'_2$, and $e'_3$ in $G'$. Let $h=|L(u)\cap L(v)|$ in $(G,L)$. If $h=0$, the triplet $e'_1, e'_2,e'_3$ is labeled with no color. If $h=1$ the triplet $e'_1, e'_2,e'_3$ is labeled with the color in common between $u$ and $v$. If $h> 1$, we create $3(h-1)$ additional parallel edges in $G'$ between $u'$ and $v'$, and we label each of the $3h$ triplets of edges with one of the colors shared by $u$ and $v$.
		Observe that $G'$ is a planar multi-graph with a given planar embedding and its treewidth is the same as the treewidth of $G$.
		
		We construct the set $D'$ of FPQ-trees associated with the vertices of $G'$ starting from the set $L$ of colors associated with the vertices of $G$ as follows. Let $v$ be a vertex of $G$, let $L(v)$ be its color list, and let $v'$ be the corresponding vertex in $G'$. Denote as $deg(v')$ the degree of $v'$ in $G'$. We equip $v'$ in $G'$ with $|L(v)|$ FPQ-trees, each encoding a color of $L(v)$ in $(G,L)$; we denote as $T_c(v')$ the FPQ-tree associated with $v'$ in $G'$ and encoding color $c\in L(v)$.
		Let $k=deg(v')/3$ (note that $k$ is a positive integer since $deg(v')\mod 3 = 0$).
		If $k = 1$, $deg(v)=1$ in $G$, and there is at most one color $c$ such that $c\in L(u) \cap L(v)$, where $u$ is the neighbor of $v$ in $G$. Each FPQ-tree $T_c(v')$ consists of a node $\rho$ whose leaves represent the triplet of edges incident to $v'$. Node $\rho$ is a Q-node if $v$ shares color $c$ with its neighbor, otherwise $\rho$ is a P-node (observe that there are at least $|L(v)|-1$ FPQ-trees associated with $v'$ with the same set of nodes).
		For example, in Figure~\ref{fi:param-c}, for vertex $z'$ we have $k=1$; the triplet of edges incident to $z'$ is labeled with color $3$, and the FPQ-tree $T_{3}(z')$ consists of a node $\rho$ with three leaves. Node $\rho$ is a Q-node because color $3\in L(z) \cap L(\tilde{z})$ in $G$.
		
		If $k > 1$, each FPQ-tree of $D'(v')$ consists of a P-node $\rho$ connected to $k$ nodes $\nu_1,\dots,\nu_{k}$ having three leaves each.
		The leaves of each $\nu_i$ ($1 \le i \le k$) represent a triplet $e'_1, e'_2, e'_3$ of edges connecting $v'$ to some other vertex $u'$ of $G'$; this triplet either encodes a color in $L(v) \cap L(u)$ or it encodes no color if $L(v) \cap L(u)=\emptyset$.
		Also, if the color $c$ associated with $T_c(v')$ is such that $c\in L(v) \cap L(u)$, node $\nu_i$ is a Q-node; it is a P-node otherwise.
		For example, in Figure~\ref{fi:param-c} we have $k>1$ for vertex $v'$. The FPQ-tree $T_{1}(v')$ encodes the color $1$ of $L(v)$; $v'$ has three triplets of incident edges and node $\rho$ of $T_{1}(v')$ has three children whose leaves represent these three triplets. Since color $1$ belongs to both $L(u)$ and $L(v)$ in Figure~\ref{fi:param-a}, the node $\nu_1$ of $T_{1}(v')$ whose leaves represent the triplet of edges $p,q,r$ of $G'$ is a Q-node. Conversely, the nodes $\nu_2$ and $\nu_3$ of $T_1(v')$ associated with the triplets labeled with colors $2$ and $3$ of $L(v)$ are P-nodes.
		
		Note that $|D'(v')|=|L(v)|$ for each vertex $v$ of $G$ and each vertex $v'$ of $G'$, thus we have that the size of $(G',D')$ is polynomial in the size of $(G,L)$.		
		We now prove that $(G,L)$ admits a proper vertex coloring with $c(v)\in L(v)$ for each $v$ if and only if $(G',D')$ is FPQ-choosable planar.
		
		Suppose first that $(G,L)$ admits a proper vertex coloring. Let $v$ be any vertex of $G$, let $c(v)$ be the chosen color for $v$, and let $v'$ be the image of $v$ in $(G',D')$. Assignment $A$ for $(G',D')$ is defined such that $A(v')=T_{c(v)}(v')$. We show that there exists a planar embedding of $G'$ that is consistent with $A$. Since any pair of adjacent vertices $u$ and $v$ in $G$ is such that $c(u) \neq c(v)$, the two FPQ-trees $A(u')=T_{c(u)}(u')$ and $A(v')=T_{c(v)}(v')$  contain pairs of nodes whose leaves correspond to triplets of edges connecting $u'$ and $v'$. Each of these triplets are connected to a P-node either in $A(u')$ or in $A(v')$ (or in both), hence they can be ordered so to avoid edge crossings in $G'$. It follows that if $(G,L)$ is a \emph{yes} instance of list coloring, then $(G',D')$ is FPQ-choosable planar.
		
		Suppose now that $(G',D')$ is FPQ-choosable planar. There exists an assignment $A$ that defines the cyclic order of the edges incident to each vertex in a planar embedding of $G'$.
		Let $u'$ and $v'$ be any two adjacent vertices of $G'$. FPQ-trees $A(v')=T_{c_1}(v')$ and $A(u')=T_{c_2}(u')$ are such that the edges represented by their leaves can be drawn in $G'$ without crossings, hence they correspond to different colors $c_1$ and $c_2$ for $v$ and $u$, and thus $c(v) \neq c(u)$.
		It follows that if $(G',D')$ is FPQ-choosable planar, then $(G,L)$ is a \emph{yes} instance of list coloring.
		It follows that \pqchoosable parameterized by treewidth is W[1]-hard.
		
%		Let $n$ ($n'$, resp.) and $m$ ($m'$, resp.) be the number of vertices and edges of $G$ ($G'$, resp.).
%		It is easy to see that $n'=n$ and $m'\le 3hm$, where $h$ is the maximum number of colors shared by any pair of adjacent vertices of $G$; the number of FPQ-trees of $D'$ is polynomial in the size of $L$ because each vertex $v'$ of $G'$ is equipped with $|L(v)|$ FPQ-trees. This implies that the size of $(G',D')$ is polynomial in the size of $(G,L)$. Moreover, the treewidth of $G'$ is the same as the treewidth of $G$, thus we can conclude that \pqchoosable parameterized by treewidth is W[1]-hard.
		
		The proof that the problem remains W[1]-hard even if the FPQ-trees associated with the vertices have only P-nodes is a slight variant of the argument above. Namely, we construct an FPQ-choosable graph $(G',D')$ such that each vertex $v$ of $G$ becomes vertex $v'$ in $G'$, and each edge $e=(u,v)$ of $G$ is split into $6$-tuples of parallel edges in $G'$.
		If $h= |L(u) \cap L(v)| =0$, the six parallel edges between $u$ and $v$ are labeled with no color. If $h=1$ the $6$-tuple of parallel edges is labeled with the color in common between $u$ and $v$.
		If $h>1$, we create $6(h-1)$ additional parallel edges between $u'$ and $v'$. Similarly tothe previous case, we label each of these $6$-tuples of edges with one of the colors in $L(u) \cap L(v)$.
		Each vertex $v'$ in $G'$ is equipped with $|L(v)|$ FPQ-trees, each encoding a color of $L(v)$.
		If $k=deg(v')/6 = 1$, each FPQ-tree $T_c(v')$ consists of a P-node $\rho$ connected to three P-nodes whose leaves represent the six edges incident to $v'$; see for example Figure~\ref{fi:pnodes-c}.
		If $k> 1$ each FPQ-tree associated with vertex $v'$ of $G'$ consists of a P-node $\rho$ connected to $k$ P-nodes $\nu_1,\dots,\nu_{k}$, each of which is connected to three P-nodes. Each of these three P-nodes has two leaves; see for example Figure~\ref{fi:pnodes-d}.		
		If $v$ shares a color $c$ with an adjacent vertex $u$, the FPQ-tree $T_c(v')$ contains a P-node $\nu_l$ ($1 \leq l \leq k$) whose leaves represent the $6$-tuple of edges connecting $v'$ with $u'$ that is labeled with color $c$. Each of these three pairs of leaves is connected to a P-node, which enforces each pair of leaves to appear consecutively. In particular, in $T_c(v')$ the two leaves $(e_i^1,e_i^2)$ must be consecutive, as well as the leaves $(e_i^3,e_i^4)$, and the leaves $(e_i^5,e_i^6)$ ($1\le i \le deg(v)$), while in $T_c(u')$ the two leaves $(e_j^5,e_j^2)$ must be consecutive, as well as the leaves $(e_j^3,e_j^6)$, and the leaves $(e_j^1,e_j^4)$ ($1\le j \le deg(u)$). This guarantees that two adjacent vertices $v'$ and $u'$ of $G'$ are such that if their FPQ-trees encode the same color $c_1$, the edges incident to $v'$ and the edges incident to $u'$ must respect cyclic orders that do not allow to connect them without edge crossings.
		On the other hand, in an FPQ-tree $T_{c_2}(u')$ encoding a color $c_2$ different from $c_1$, the pairs of leaves that must be consecutive are the same as the ones of $T_{c_1}(v')$, which allows to connect the corresponding edges of $G'$ without edge crossings.
		By this argument, we can conclude that the \pqchoosable is W[1]-hard even if the FPQ-trees associated with the vertices have only P-nodes.
		\qed
		
				%		\begin{figure}[tb]
		%			\centering
		%			\includegraphics[width=.7\textwidth,page=5]{param}
		%			\caption{Two FPQ-trees associated with a same color that is shared by two adjacent vertices of $G'$. The common leaves are $\ell_1$, $\ell_2$, $\ell_3$, $\ell_4$, $\ell_5$, $\ell_6$.}
		%			\label{fi:w1-pnodes}
		%		\end{figure}
		
		\begin{figure}[tbp]
			\centering
			\subfigure[]{\includegraphics[width=.12\textwidth,page=10]{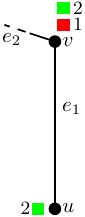}\label{fi:pnodes-a}}
			\hfill
			\subfigure[]{\includegraphics[width=.21\textwidth,page=11]{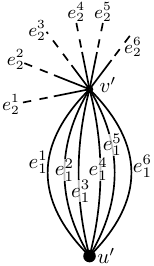}\label{fi:pnodes-b}}
			\hfill
			\subfigure[]{\includegraphics[width=.3\textwidth,page=5]{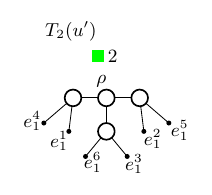}\label{fi:pnodes-c}}
			\hfill			\subfigure[]{\includegraphics[width=.3\textwidth,page=9]{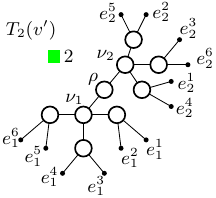}\label{fi:pnodes-d}}
			\caption{(a) An edge $(u,v)$ of an instance of list coloring. (b) The corresponding FPQ-choosable graph. (c) The FPQ-tree $T_2(u')$ of $D'(u')$ associated with color $2$. (d) The FPQ-tree $T_2(v')$ of $D'(v')$ associated with color $2$. Note that they contain only P-nodes.}
			\label{fi:pnodes}
		\end{figure}
	\end{proof}

	The results of this section imply the following.

	\begin{corollary}\label{co:no-fpt}
		\pqchoosable is not fixed-parameter tractable if parameterized by treewidth only or by the number of FPQ-trees per vertex only. It remains fixed-parameter tractable even if the FPQ-trees consist of P-nodes.
	\end{corollary}

%Theorem~\ref{th:pqchoosable-npcomplete} naturally raises the question about which families of FPQ-choosable graphs other than those studied by Gutwenger et al. admit a polynomial-time solution for the \pqchoosable problem. The next section proves that any FPQ-choosable graph $(G,D)$ such that $G$ has bounded branchwidth and $|D(v)|$ is bounded by a constant for each vertex $v$ can be tested for FPQ-choosable planarity in polynomial time.

%Theorem~\ref{th:pqchoosable-npcomplete} motivates the design of fixed-parameter tractable solutions for the \pqchoosable problem when the number of FPQ-trees associated with each vertex is bounded by a constant. We are going to prove that an FPQ-choosable graph $(G,D)$ such that $G$ has bounded branchwidth and $|D(v)|$ is bounded by a constant for each vertex $v$ of $G$ can be tested for FPQ-choosable planarity in polynomial time.

\section{Fixed Parameter Tractability of FPQ-choosable Planarity Testing}\label{se:pqchoosable-testing-fpt}
				
This section is organized as follows. We first introduce the notions of boundaries and of extensible orders, and state two technical lemmas. Next, we define the concepts of pertinent FPQ-tree, skeletal FPQ-tree and admissible tuple, which are fundamental in the algorithm description. Finally, we present a polynomial-time testing algorithm for FPQ-choosable graphs having bounded branchwidth and such that the number of FPQ-trees associated with each vertex is bounded by a constant.
Note that, if a graph has bounded branchwidth $b$ it has treewidth at most $\big \lfloor {\frac{3}{2} b} \big \rfloor -1$~\cite{rs-gm-91}.

\smallskip	
\noindent \textbf{Boundaries and Extensible Orders:} Let $T$ be an FPQ-tree, let $\yield(T)$ denote the set of its leaves, and let $L$ be a proper subset of $\yield(T)$. We denote by $\sigma$ a cyclic order of the leaves of an FPQ-tree, and we say that $\sigma \in \consistent(T)$ if the FPQ-tree $T$ represents $\sigma$. We say that $L$ is a \emph{consecutive set} if the leaves in $L$ are consecutive in every cyclic order represented by $T$. Let $e$ be an edge of $T$, and let $T'$ and $T''$ be the two subtrees obtained by removing $e$ from $T$. If either $\yield(T')$ or $\yield(T'')$ are a subset of a consecutive set $L$, then we say that $e$ is a \emph{split edge of $L$}. The subtree that contains the leaves in $L$ is the \emph{split subtree} of $e$.
A split edge $e$ is \emph{maximal} if there exists no split edge $e'$ such that the split subtree of $e'$ contains $e$.

	\begin{lemma}\label{le:boundary}
		Let $T$ be an FPQ-tree, let $L$ be a consecutive proper subset of $\yield(T)$, and let $S$ be the set of maximal split edges of $L$. Then either $|S|=1$, or $|S|>1$ and there exists a Q-node (or an F-node) $\chi$ of $T$ such that $\chi$ has degree at least $|S|+2$ and the elements of $S$ appear consecutive around $\chi$.
	\end{lemma}
	\begin{proof}
		\begin{figure}[tbp]
			\centering
			{\includegraphics[width=.4\textwidth,page=1]{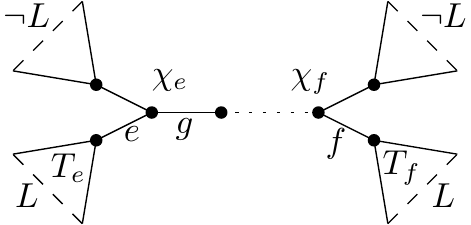}}
			\caption{Illustration for the proof of Lemma~\ref{le:boundary}.}
			\label{fi:boundary}
		\end{figure}
		Assume that $|S| > 1$. Let $e$ and $f$ be two maximal split edges
		of $L$, and let $T_e$ and $T_f$ be the split subtrees of $e$
		and $f$, respectively. Let further $\chi_e$ denote the
		endpoint of $e$ that is not in $T_e$. The endpoint $\chi_f$
		is defined likewise. Refer to Figure~\ref{fi:boundary} for an illustration.		
		
		Suppose for the sake of contradiction that~$\chi_e$
		and~$\chi_f$ are distinct. Let~$g$ denote the
		first edge on the path from~$\chi_e$ to~$\chi_f$.  By the
		maximality of~$e$ and~$f$, the edge $g$ is not a split edge.
		It follows that there is an edge $e'$ incident to~$\chi_e$
		that is different from~$g$ and that is not a split edge.
		Likewise, we find an edge $f'$ incident to~$\chi_f$ that is
		different from the first edge on the path from~$\chi_f$
		to~$\chi_e$ and that is not a split edge.  But then~$g$ is
		an edge of a tree $T$ such that one of the two subtrees it
		separates has leaves in~$L$ and leaves that are not in~$L$.
		It follows that~$L$ is not a consecutive set. This is a
		contradiction to the assumption that~$\chi_e$ and~$\chi_f$
		are distinct.
		
		It follows that the edges in $S$ are all incident to a
		single vertex~$\chi$. If~$\chi$ has degree~$|S|$,
		then~$L$ is not a proper subset of the leaves, and if it has
		degree~$|S|+1$, then also its remaining edge is a split
		edge, which contradicts the maximality of the split edges
		in~$S$.  Hence $\deg(\chi) \ge |S|+2$.  If~$\chi$ were a
		P-node, this would contradict the assumption that~$L$ is a
		consecutive set.
		\qed
	\end{proof}

	If $|S|=1$, the split edge in $S$ is called the \emph{boundary of L}. If $|S|>1$, the Q-node (or the F-node) $\chi$ defined in the statement of Lemma~\ref{le:boundary} is the \emph{boundary of $L$}.
	Since F-nodes are a more constrained version of Q-nodes, when we refer to boundary Q-nodes we also take into account the case in which they are F-nodes.
	Figure~\ref{fi:boundary-Snode-a} shows an FPQ-choosable graph $(G,D)$ and
	two FPQ-trees $T_u \in D(u)$ and $T_v \in D(v)$. The three red edges $b$, $c$, and $d$ of $G$ define a consecutive set $L_u$ in $T_u$; the edges $e$ and $f$ define a consecutive set $L_v$ in $T_v$. The boundary of $L_u$ in $T_u$ is a Q-node, while the boundary of $L_u$ in $T_u$ is an edge.
	We denote as $\mathcal{B}(L)$ the boundary of a set of leaves $L$. If $\mathcal{B}(L)$ is a Q-node, we associate $\mathcal{B}(L)$ with a default orientation (i.e., a flip) that arbitrarily defines one of the two possible permutations of its children. We call this default orientation the \emph{clockwise orientation} of $\mathcal{B}(L)$. The other possible permutation of the children of $\mathcal{B}(L)$ corresponds to the \emph{counter-clockwise orientation}.

	Let $L'=L \cup \{\ell\}$, where $\ell$ is a new element. Let $\sigma \in \consistent(T)$, and let $\sigma|_{L'}$ be a cyclic order obtained from $\sigma$ by replacing the elements of the consecutive set $\yield(T) \setminus L$ by the single element $\ell$. We say that a cyclic order $\sigma'$ of $L'$ is \emph{extensible} if there exists a cyclic order $\sigma \in \consistent(T)$ with $\sigma|_{L'}=\sigma'$. In this case, we say that $\sigma$ is an \emph{extension of $\sigma'$}. 
	Note that if the boundary of $L$ is a Q-node $\chi$, then any two extensions of $\sigma'$ induce the same clockwise or counter-clockwise orientation of the edges incident to $\chi$. An extensible order $\sigma$ is \emph{clockwise} if the orientation of $\chi$ is clockwise; $\sigma$ is \emph{counter-clockwise} otherwise.
	If the boundary of $L$ is an edge, we consider any extensible order as both clockwise and counter-clockwise.
	
	\begin{figure}[tb]
		\centering
		\subfigure[]{\includegraphics[width=.7\textwidth,page=3]{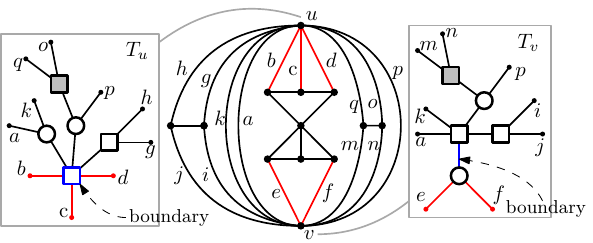}\label{fi:boundary-Snode-a}}
		\\
		\subfigure[]{\includegraphics[width=.7\textwidth,page=4]{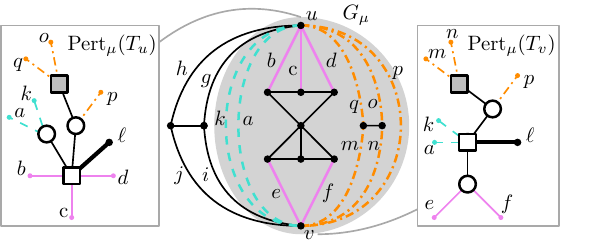}\label{fi:boundary-Snode-b}}
		\hfil
		\subfigure[]{\includegraphics[width=.28\textwidth,page=5]{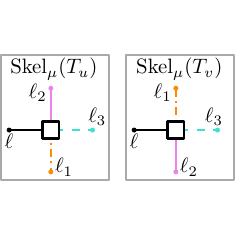}\label{fi:boundary-Snode-c}}
		\caption{(a) Two different types of boundaries: A boundary Q-node in $T_u$ and a boundary edge in $T_v$. (b) The pertinent FPQ-trees $\pert_{\mu}(T_u)$ of $T_u$ and $\pert_{\mu}(T_v)$ of $T_v$. (c) The skeletal FPQ-trees $\skeletal_{\mu}(T_u)$ of $\pert_{\mu}(T_u)$ and $\skeletal_{\mu}(T_v)$ of $\pert_{\mu}(T_v)$.}
		\label{fi:boundary-Snode}
	\end{figure}
	
	Let $L$ and $\hat{L}$ be two disjoint consecutive sets of leaves that have the same boundary Q-node $\chi$ in $T$. 
	Let $\sigma$ and $\hat{\sigma}$ be two extensible orders of $L$ and $\hat{L}$, respectively. We say that $\sigma$ and $\hat{\sigma}$ are \emph{incompatible} if one of them is clockwise and the other one is counter-clockwise.

	\begin{lemma}\label{le:incompatibility}
		Let $T$ be an FPQ-tree, let $L_1 \cup \dots \cup L_k$ be a partition of $\yield(T)$ into consecutive sets, and let $\sigma_1, \dots, \sigma_k$ be extensible orders of $L_1, \dots, L_k$. There exists an order $\Sigma$ of $\yield(T)$ represented by $T$ such that $\Sigma|_{L_i}=\sigma_i$ if and only if no pair $\sigma_i,\sigma_j$ $(1 \le i, j \le k)$ is incompatible.
	\end{lemma}
	\begin{proof}
	  The only-if direction is clear.  For the if-direction, assume that no pair is incompatible.  Note that, since~$L_i$ is consecutive, so is~$\yield(T) \setminus L_i$.  We denote by~$T_i$ the subtree of~$T$ that is obtained by replacing the consecutive set~$\yield(T) \setminus L_i$ by a single leaf~$\ell$.  Note that~$T_i$ ($1 \le i \le k$) is a subtree of~$T$ and the set $\{T_1, \dots, T_k\}$ forms a partition of the edges of~$T$. 
	  Observe that~$\sigma_i$ defines a cyclic order of the edges around each node of~$T_i$.  Moreover, if~$T_i$ and~$T_j$ overlap, then they do so in
		the boundary of~$L_i$ and~$L_j$, which must hence be a
		Q-node~$\chi$.  Since no pair is incompatible, it follows that they
		induce the same cyclic order $\Sigma$ of the edges around~$\chi$.  Thus,
		together the~$\sigma_i$ determine a unique order in
		$\consistent(T)$ such that $\Sigma|_{L_i}=\sigma_i$.
		\qed
	\end{proof}
	
\smallskip	
\noindent \textbf{Pertinent FPQ-trees, Skeletal FPQ-trees, and Admissible Tuples:}
	Let $(G,D)$ be an FPQ-choosable graph, let $\mathcal{T}$ be an SPQR-decomposition tree of $G$ and let $v$ be a pole of a node $\mu$ of $\mathcal{T}$, let $T_v\in D(v)$ be an FPQ-tree associated with $v$, let $E_\mathrm{ext}$ be the set of edges that are incident to $v$ and not contained in $G_\mu$, and let $E_{\mu}^\star(v)=E(v)\setminus E_\mathrm{ext}$. Note that there is a bijection between the edges $E(v)$ of $G$ and the leaves of $T_v$, hence we shall refer to the set of leaves of $T_v$ as $E(v)$. Also note that $E_{\mu}^\star(v)$ is represented by a consecutive set of leaves in $T_v$, because in every planar embedding of $G$ the edges in $E_{\mu}^\star(v)$ must appear consecutively in the cyclic order of the edges incident to $v$.

    The \emph{pertinent FPQ-tree} of $T_v$, denoted as $\pert_\mu(T_v)$, is the FPQ-tree obtained from $T_v$ by replacing the consecutive set $E_\mathrm{ext}$ with a single leaf $\ell$. Informally, the pertinent FPQ-tree of $v$ describes the hierarchical embedding constraints for the pole $v$ within the pertinent graph $G_\mu$. For example, in Figure~\ref{fi:boundary-Snode-b} a pertinent graph $G_\mu$ with poles $u$ and $v$ is highlighted by a shaded region; the pertinent FPQ-tree $\pert_{\mu}(T_u)$ of $T_u$ and the pertinent FPQ-tree $\pert_{\mu}(T_v)$ of $T_v$ are obtained by the FPQ-trees $T_u$ and $T_v$ of Figure~\ref{fi:boundary-Snode-a}.
 
    Let $\nu_1, \dots, \nu_k$ be the children of $\mu$ in $\mathcal{T}$. Observe that the edges $E_{\nu_i}^\star(v)$ of each $G_{\nu_i}$ ($1 \le i \le k$) form a consecutive set of leaves of $A_\mu(v)=\pert_{\mu}(T_v)$.
    The \emph{skeletal FPQ-tree} of $\pert_{\mu}(T_v)$, denoted by $\skeletal_\mu(T_v)$, is the tree obtained from $\pert_\mu(T_v)$ by replacing each of the consecutive sets~$E_{\nu_i}^\star(v)$ ($1 \le i \le k$) by a single leaf~$\ell_i$. See for example, Figure~\ref{fi:boundary-Snode-c}.
    Observe that each
    Q-node of $\skeletal_\mu(T_u)$ corresponds to a Q-node of $\pert_{\mu}(T_u)$, and thus to a Q-node of $T_u$; also, distinct Q-nodes of~$\skeletal_\mu(T_u)$ correspond to distinct Q-nodes of
    $\pert_{\mu}(T_u)$, and thus to distinct Q-nodes of $T_u$.  For each Q-node~$\chi$
    of~$T_u$ that is a boundary of~$\mu$ or of one of its
    children~$\nu_i$, there is a corresponding Q-node in~$\skeletal_\mu(T_u)$ that inherits its default orientation from $T_u$.
%    
%    Let $\emb_{\mu}$ be a planar embedding of $G_\mu$ consistent with $A_\mu$.  Let $\nu_1, \dots, \nu_k$ be the children of $\mu$ in $\mathcal{T}$. Observe that, by the planarity of $\emb_{\mu}$, the edges $E_{\nu_i}^\star(u)$  of each $G_{\nu_i}$ ($1 \le i \le k$) form a consecutive set of leaves of $A_\mu(u)=\pert_{\mu}(T_u)$.
%	The \emph{skeletal FPQ-tree} of $\pert_{\mu}(T_u)$ in $\emb_{\mu}$, denoted by $\skeletal_\mu(T_u)$, is the tree obtained from $\pert_\mu(T_u)$ by replacing each
%	of the consecutive sets~$E_{\nu_i}^\star(u)$ ($1 \le i \le k$) by a single leaf~$\ell_i$.
	
%	\begin{figure}[tb]
%		\centering
%%		\subcaptionbox{\label{fi:trees-a}}
%%		{\includegraphics[width=.3\textwidth,page=1]{trees}}
%%		\hfill
%		\subcaptionbox{\label{fi:trees-b}}
%		{\includegraphics[width=.33\textwidth,page=2]{trees}}
%		\hfil
%		\subcaptionbox{\label{fi:trees-c}}
%		{\includegraphics[width=.33\textwidth,page=3]{trees}}
%		\caption{(a) The pertinent FPQ-tree $\pert_{\mu}(T_v)$ of $T_v$. (b) The skeletal FPQ-tree $\skeletal_\mu(T_v)$ of $\pert_{\mu}(T_v)$.}
%		\label{fi:trees}
%	\end{figure}

	Let $(G,D)$ be an FPQ-choosable graph, let $\mathcal{T}$ be an SPQR-decomposition tree of $G$, let $\mu$ be a node of $\mathcal{T}$, and let $u$ and $v$ be the poles of $\mu$. 
	We denote with $(G_\mu,D_\mu)$ the FPQ-choosable graph consisting of the pertinent graph $G_\mu$ and the set $D_\mu$ that is defined as follows:
	$D_\mu(z)=D(z)$ for each vertex $z$ of $G_\mu$ that is not a pole, and $D_\mu(v)=\{\pert_\mu(T_v)\mid T_v\in D(v)\}$ if $v$ is a pole of $\mu$.
	A tuple $\langle T_u,T_v,o_u,o_v\rangle \in D(u) \times D(v) \times \{0,1\} \times \{0,1\}$ is \emph{admissible for $G_\mu$} if there exists an assignment $A_\mu$ of $(G_\mu,D_\mu)$ and a planar embedding $\emb_\mu$ of $G_\mu$ consistent with $A_\mu$ such that $A_\mu(u)=\pert_{\mu}(T_u)$, $A_\mu(v)=\pert_{\mu}(T_v)$, $\mathcal{B}(E_{\mu}^\star(u))$ is clockwise (counter-clockwise) in $T_u$ if $o_u=0$ ($o_u=1$), and $\mathcal{B}(E_{\mu}^\star(v))$ is clockwise (counter-clockwise) in $T_v$ if $o_v=0$ ($o_v=1$).
	We say that a tuple is \emph{admissible for $\mu$} if it is admissible for $G_\mu$.
	We denote by $\Psi(\mu)$ the set of admissible tuples for $G_\mu$.
	% $\emb_\mu(u)$ is clockwise in $T_u$ if $o_u=0$ NOT IF AND ONLY IF, BECAUSE IT IS NOT CORRECT FOR THE CASE IN WHICH THE BOUNDARY IS AN EDGE

\smallskip

\noindent \textbf{FPT Algorithm:} In order to test if $(G,D)$ is FPQ-choosable planar, we root the SPQR-decomposition tree $\mathcal{T}$ at an arbitrary Q-node and we visit $\mathcal{T}$ from the leaves to the root. At each step of the visit, we equip the currently visited node $\mu$ with the set $\Psi(\mu)$. If we encounter a node $\mu$ such that $\Psi(\mu)= \emptyset$, we return that $(G,D)$ is not FPQ-choosable planar; otherwise the planarity test returns an affirmative answer.
If the currently visited node $\mu$ is a leaf of $\mathcal{T}$, we set $\Psi(\mu)=D(u) \times D(v) \times \{0,1\} \times \{0,1\}$, because its pertinent graph is a single edge.
If $\mu$ is an internal node, $\Psi(\mu)$ is computed from the sets of admissible tuples of the children of $\mu$. The next lemmas describe how to compute $\Psi(\mu)$ depending on whether $\mu$ is an S-, P-, or R-node.
	
	\begin{lemma}\label{le:Snode}
		Let $\mu$ be an S-node with children $\nu_1$ and $\nu_2$. Given $\Psi(\nu_1)$ and $\Psi(\nu_2)$, the set $\Psi(\mu)$ can be computed in $O(D_{\mathrm{max}}^2 \log (D_{\mathrm{max}}))$ time, where $D_{\mathrm{max}}= \max_{v\in V}|D(v)|$.
	\end{lemma}
	\begin{proof}
		Let $u$ and $v$ be the poles of $\mu$, and let $w$ be the pole in common between $G_{\nu_1}$ and $G_{\nu_2}$. We show that $\langle T_u,T_v, o_u, o_v \rangle \in \Psi(\mu)$ if and only if there exist a tree $T_w \in D(w)$ and an $o_w \in \{0,1\}$, such that $\langle T_u, T_w, o_u, o_w \rangle \in \Psi(\nu_1)$ and $\langle T_w, T_v, o_w, o_v \rangle \in \Psi(\nu_2)$.
		
		If $\langle T_u, T_v, o_u, o_v \rangle \in \Psi(\mu)$, then there exist an assignment $A_\mu$ of $(G_\mu, D_\mu)$ and a planar embedding $\emb_\mu$ of $G_\mu$ consistent with $A_\mu$ such that $A_\mu(u)=\pert_\mu(T_u)$, $A_\mu(v)=\pert_\mu(T_v)$, $\mathcal{B}(E_{\mu}^\star(u))$ is clockwise (counter-clockwise) in $T_u$ if $o_u=0$ ($o_u=1$), and $\mathcal{B}(E_{\mu}^\star(v))$ is clockwise (counter-clockwise) in $T_v$ if $o_v=0$ ($o_v=1$). Let $\emb_{\nu_1}$ and $\emb_{\nu_2}$ be the embeddings of $G_{\nu_1}$ and $G_{\nu_2}$ induced by $\emb_\mu$, respectively, and let $T_w = A_\mu(w)$. Observe that $E_{\nu_1}^\star(w)$ and $E_{\nu_2}^\star(w)$ are disjoint consecutive sets of $T_w$ sharing the same boundary in $T_w$.
		Also, observe that $\emb_{\mu}(w)$ is an extension of both $\emb_{\nu_1}(w)$ and $\emb_{\nu_2}(w)$. By Lemma~\ref{le:incompatibility}, $\emb_{\nu_1}(w)$ and $\emb_{\nu_2}(w)$ are not incompatible, and hence $\mathcal{B}(E_{\nu_1}^\star(w))$ and $\mathcal{B}(E_{\nu_2}^\star(w))$ are both clockwise or both counter-clockwise.
		We set $o_w=0$ if they are both clockwise, and $o_w=1$ otherwise.
%		
%		We define an assignment $A_{\nu_1}$ for $(G_{\nu_1}, D_{\nu_1})$ as follows.
		For every vertex $x$ of $G_{\nu_1}$ different from $w$, we set $A_{\nu_1}(x)=A_{\mu}(x)$; for $w$ we set $A_{\nu_1}(w)=\pert_{\nu_1}(T_w)$.
		Since $\emb_{\nu_1}$ is consistent with $A_{\nu_1}$ and $\emb_{\nu_1}(u)=\emb_{\mu}(u)$, $\mathcal{B}(E_{\nu_1}^\star(u))$ is clockwise (counter-clockwise) in $T_u$ if $o_u=0$ (if $o_u=1$).
		By observing that $\mathcal{B}(E_{\nu_1}^\star(w))$ is clockwise (counter-clockwise) if $o_w=0$ (if $o_w=1$), we have that $\langle T_u, T_w, o_u, o_w \rangle \in \Psi(\nu_1)$.
		The same argument can be used to show that $\langle T_u, T_w, o_u, o_w \rangle \in \Psi(\nu_2)$.
		It follows that if $\langle T_u,T_v, o_u, o_v \rangle \in \Psi(\mu)$, there exist a tree $T_w \in D(w)$ and an $o_w \in \{0,1\}$ such that $\langle T_u, T_w, o_u, o_w \rangle \in \Psi(\nu_1)$ and $\langle T_w, T_v, o_w, o_v \rangle \in \Psi(\nu_2)$.
		
		For the converse, assume that there exist a tree $T_w \in D(w)$ and an $o_w \in \{0,1\}$, such that $\theta_1=\langle T_u, T_w, o_u, o_w \rangle \in \Psi(\nu_1)$ and $\theta_2=\langle T_w, T_v, o_w, o_v \rangle \in \Psi(\nu_2)$. By definition, there exist assignments $A_{\nu_1}$ and $A_{\nu_2}$ of $(G_{\nu_1},D_{\nu_1})$ and $(G_{\nu_2},D_{\nu_2})$ respectively, and two planar embeddings $\emb_{\nu_1}$ and $\emb_{\nu_2}$ that are consistent with $A_{\nu_1}$ and $A_{\nu_2}$ respectively, 
		%		We can construct a planar embedding of $\emb_{\mu}$ by a series composition of $\emb_{\nu_1}$ and $\emb_{\nu_2}$ by making $w$ coincide.
		%
%%%%%%%%%%%%%%%%% ALESSANDRA %%%%%%%%%%%%%%%%%%%%%%%%%%%
		such that $A_{\nu_1}(u)=\pert_{\nu_1}(T_u)$, $A_{\nu_1}(w)=\pert_{\nu_1}(T_w)$, $\mathcal{B}(E_{\nu_1}^\star(u))$ is clockwise (counter-clockwise) in $T_u$ if $o_u=0$ ($o_u=1$), $\mathcal{B}(E_{\nu_1}^\star(w))$ is clockwise (counter-clockwise) in $T_w$ if $o_w=0$ ($o_w=1$), $A_{\nu_2}(w)=\pert_{\nu_2}(T_w)$, $A_{\nu_2}(v)=\pert_{\nu_2}(T_v)$, $\mathcal{B}(E_{\nu_2}^\star(w))$ is clockwise (counter-clockwise) in $T_w$ if $o_w=0$ ($o_w=1$), and $\mathcal{B}(E_{\nu_2}^\star(v))$ is clockwise (counter-clockwise) in $T_v$ if $o_v=0$ ($o_v=1$).
		We define an assignment $A_\mu$ and a planar embedding $\emb_{\mu}$ of $G_\mu$ consistent with $A_\mu$ such that $A_\mu(u)=\pert_{\mu}(T_u)$, $A_\mu(v)=\pert_{\mu}(T_v)$, $\mathcal{B}(E_{\mu}^\star(u))$ is clockwise (counter-clockwise) in $T_u$ if $o_u=0$ ($o_u=1$), and $\mathcal{B}(E_{\mu}^\star(v))$ is clockwise (counter-clockwise) in $T_v$ if $o_v=0$ ($o_v=1$).
		Embedding $\emb_{\mu}$ of $G_\mu$ is obtained by merging $\emb_{\nu_1}$ and $\emb_{\nu_2}$ as follows.
		For every vertex $x$ of $G_{\nu_1}$ different from $w$, we set $\emb_{\mu}(x)=\emb_{\nu_1}(x)$, for every vertex $y$ of $G_{\nu_2}$ different from $w$, we set $\emb_{\mu}(y)=\emb_{\nu_2}(y)$. For $w$, since $o_w$ has the same value in $\theta_1$ and in $\theta_2$, hence $\mathcal{B}(E_{\nu_1}^\star(w))$ and $\mathcal{B}(E_{\nu_2}^\star(w))$ are not incompatible. By Lemma~\ref{le:incompatibility}, there exists an order of the leaves of $T_w$ that is an extension of both $\emb_{\nu_1}(w)$ and $\emb_{\nu_2}(w)$: Let $\emb_{\mu}(w)$ be this order.
		Assignment $A_{\mu}$ for $(G_{\mu}, D_{\mu})$ is defined as follows.
		For every vertex $x$ of $G_{\nu_1}$ different from $w$, we set $A_{\mu}(x)=A_{\nu_1}(x)$; for every vertex $y$ of $G_{\nu_2}$ different from $w$, we set $A_{\mu}(y)=A_{\nu_2}(y)$; for $w$ we set $A_{\mu}(w)=T_{w}$.
		Since $\emb_{\mu}$ is consistent with $A_{\mu}$, $\emb_{\mu}(u)=\emb_{\nu_1}(u)$, and $\emb_{\mu}(v)=\emb_{\nu_2}(v)$, $\mathcal{B}(E_{\mu}^\star(u))$ is clockwise (counter-clockwise) in $T_u$ if $o_u=0$ (if $o_u=1$), and $\mathcal{B}(E_{\mu}^\star(v))$ is clockwise (counter-clockwise) in $T_v$ if $o_v=0$. Furthermore, since $E_{\nu_1}^\star(u)=E_{\mu}^\star(u)$ and $E_{\nu_2}^\star(v)=E_{\mu}^\star(v)$, $A_\mu(u)=\pert_{\mu}(T_u)$ and $A_\mu(v)=\pert_{\mu}(T_v)$. It follows that if there exist a tree $T_w \in D(w)$ and an $o_w \in \{0,1\}$ such that $\langle T_u, T_w, o_u, o_w \rangle \in \Psi(\nu_1)$ and $\langle T_u, T_w, o_u, o_w \rangle \in \Psi(\nu_2)$, then $\langle T_u,T_v, o_u, o_v \rangle \in \Psi(\mu)$.
		
%%%%%%%%%%%%%%%%%%%%%%%%%%%%%%%%%%%%%%%%%%%%%%%%%%%%%%%
		
%		Let $G_{\nu_1}$ and $G_{\nu_2}$ be the pertinent graphs of $\nu_1$ and $\nu_2$ respectively, and let $w$ be the pole in common between $G_{\nu_1}$ and $G_{\nu_2}$. A tuple $\langle T_u,T_v, o_u, o_v \rangle$ is admissible for $G_\mu$ if and only if there exists a $T_w \in D(w)$ with orientation $o_w$ for its boundary such that $\langle T_u, T_w, o_u, o_w \rangle$ is an admissible tuple for $G_{\nu_1}$  and $\langle T_w, T_v, o_w, o_v \rangle$ is an admissible tuple for $G_{\nu_2}$.
		
		Set $\Psi(\mu)$ is computed from $\Psi(\nu_1)$ and $\Psi(\nu_2)$ by looking for pairs of tuples $\langle T_u, T_w, o_u, o_w \rangle \in \Psi(\nu_1)$, $\langle T_w, T_v, o_w, o_v \rangle \in \Psi(\nu_2)$ sharing the same $T_w$ and the same value of $o_w$.
		By ordering $\Psi(\nu_1)$ and $\Psi(\nu_2)$, $\Psi(\mu)$ is computed in $O(D_{\mathrm{max}}^2 \log (D_{\mathrm{max}}))$ time.
		\qed
	\end{proof}

	\begin{lemma}\label{le:Pnode}
		Let $\mu$ be a P-node with children $\nu_1,\nu_2,\dots,\nu_k$. Given $\Psi(\nu_1),\Psi(\nu_2),$ $\dots,\Psi(\nu_k)$, the set $\Psi(\mu)$ can be computed in $O(D_{\mathrm{max}}^2 \cdot n)$ time, where $D_{\mathrm{max}} = \max_{v\in V}|D(v)|$.
	\end{lemma}

	\begin{proof}

		Let~$u$ and~$v$ be the poles of~$\mu$. Let~$\skeletal_\mu(T_u)$ and $\skeletal_\mu(T_v)$ be the skeletal FPQ-trees of $\pert_{\mu}(T_u)$ and of $\pert_{\mu}(T_v)$, respectively. It can be proved that a tuple $\langle T_u, T_v, o_u, o_v \rangle \in \Psi(\mu)$ if and only if the following two conditions are satisfied:
		(i) There exists a planar embedding $\emb_{\mu}$ of $\skeleton(\mu)$ and a pair of skeletal FPQ-trees $\skeletal_{\mu}(T_u)$ and $\skeletal_{\mu}(T_v)$ such that $\emb_{\mu}(u)\in \consistent(\skeletal_{\mu}(T_u))$ and $\emb_{\mu}(v)\in \consistent(\skeletal_{\mu}(T_v))$; (ii) For each child $\nu_i$ of $\mu$ ($1 \le i \le k$), there exist an orientation $o_u$ of $\mathcal{B}(E_{\nu_i}^\star(u))$ and an orientation $o_v$ of $\mathcal{B}(E_{\nu_i}^\star(v))$ such that $\langle T_u, T_v, o_u, o_v \rangle \in \Psi(\nu_i)$.

		Let~$u$ and~$v$ be the poles of~$\mu$. Let~$\skeletal_\mu(T_u)$ and $\skeletal_\mu(T_v)$ be the skeletal FPQ-trees of $\pert_{\mu}(T_u)$ and of $\pert_{\mu}(T_v)$, respectively. We first show that a tuple $\langle T_u, T_v, o_u, o_v \rangle \in \Psi(\mu)$ if and only if the following two conditions are satisfied:
		\begin{enumerate}[(i)]
			\item There exists a planar embedding $\emb_{\mu}$ of $\skeleton(\mu)$ and a pair of skeletal FPQ-trees $\skeletal_{\mu}(T_u)$ and $\skeletal_{\mu}(T_v)$ such that $\emb_{\mu}(u)\in \consistent(\skeletal_{\mu}(T_u))$ and $\emb_{\mu}(v)\in \consistent(\skeletal_{\mu}(T_v))$;
			
			\item For each child $\nu_i$ of $\mu$ ($1 \le i \le k$), there exist an orientation $o_u$ of $\mathcal{B}(E_{\nu_i}^\star(u))$ and an orientation $o_v$ of $\mathcal{B}(E_{\nu_i}^\star(v))$ such that $\langle T_u, T_v, o_u, o_v \rangle \in \Psi(\nu_i)$.
		\end{enumerate}
		If $\langle T_u, T_v, o_u, o_v \rangle \in \Psi(\mu)$, then there exist an assignment $A_\mu$ of $(G_\mu, D_\mu)$ and a planar embedding $\emb_\mu$ of $G_\mu$ consistent with $A_\mu$. Let $\skeletal_{\mu}(T_u)$ and $\skeletal_{\mu}(T_v)$  be the skeletal FPQ-trees obtained from $A_\mu(u)$ and from $A_\mu(v)$, respectively. By definition of skeletal FPQ-tree, the planar embedding $\emb_\mu$ and the pair of skeletal FPQ-trees $\skeletal_{\mu}(T_u)$ and $\skeletal_{\mu}(T_v)$ satisfy Condition~(i).  Let $\emb_{\nu_i}$ be the embedding of $G_{\nu_i}$ induced by $\emb_{\mu}$ ($1 \le i \le k$). $E_{\nu_i}^\star(u)$ is a consecutive set of $\pert_{\mu}(T_u)$ and $E_{\nu_i}^\star(v)$ is a consecutive set of $\pert_{\mu}(T_v)$. Note that $\emb_{\mu}(u)$ is an  extension of $\emb_{\nu_i}(u)$ and that $\emb_{\mu}(v)$ is an extension of $\emb_{\nu_i}(v)$. We can therefore define an assignment $A_{\nu_i}$ for $(G_{\nu_i},D_{\nu_i})$ as follows: For every vertex $w$ of $G_{\nu_i}$ different from the poles of $G_{\nu_i}$, we set $A_{\nu_i}(w)= A_\mu(w)$; for the poles of $G_{\nu_i}$ we set $A_{\nu_i}(u)=\pert_{\nu_i}(T_u)$ and $A_{\nu_i}(v)=\pert_{\nu_i}(T_v)$.
		Note that $\emb_{\nu_i}$ is consistent with $A_{\nu_i}$. Thus, there exist values 
		$o_u \in \{0,1\}$ and $o_v \in \{0,1\}$ such that $\langle T_u, T_v, o_u, o_v \rangle \in \Psi(\nu_i)$ and hence Condition~(ii) is satisfied. It follows that if $\langle T_u, T_v, o_u, o_v \rangle \in \Psi(\mu)$, both Condition~(i) and Condition~(ii) are satisfied. 
		
		Suppose now that Condition~(i) and Condition~(ii) are satisfied. By Condition~(i), the planar embedding $\emb_\mu$ and the pair of skeletal FPQ-trees $\skeletal_{\mu}(T_u)$ and $\skeletal_{\mu}(T_v)$ describe how to arrange the children around $u$ and $v$ in a planar embedding of $\skeleton(\mu)$, since the union of all $E_{\nu_i}^\star(u)$ coincides with $E_{\mu}^\star(u)$ and the union of all $E_{\nu_i}^\star(v)$ coincides with $E_{\mu}^\star(v)$ ($1 \le i \le k$). By Condition~(ii) there exist an assignment $A_{\nu_i}$ of $(G_{\nu_i},D_{\nu_i})$ and a planar embedding $\emb_{\nu_i}$ that is consistent with $A_{\nu_i}$. A planar embedding $\emb_{\mu}$ of $G_\mu$ is obtained by merging all the $\emb_{\nu_i}$. More precisely, for every vertex $w$ of $G_{\nu_i}$ different from the poles, we set $\emb_{\mu}(w)=\emb_{\nu_i}(w)$. Concerning the poles $u$ and $v$, observe that there exists an order of the leaves of $T_u$ that is a common extension of all $\emb_{\nu_i}(u)$, and an order of the leaves of $T_v$ that is a common extension of all $\emb_{\nu_i}(v)$: Let $\emb_{\mu}(u)$ and $\emb_{\mu}(v)$ be these orders. Also, for every vertex $w$ of $G_{\nu_i}$ different from $u$ and $v$, we set $A_{\mu}(w)=A_{\nu_i}(w)$. For the poles $u$ and $v$ we set $A_{\mu}(u)=\pert_\mu(T_{u})$ and $A_{\mu}(v)=\pert_\mu(T_{v})$, respectively. Thus obtaining an embedding $\emb_{\mu}$ that is consistent with $A_{\mu}$. It follows that if Condition~(i) and Condition~(ii) are satisfied, then $\langle T_u, T_v, o_u, o_v \rangle \in \Psi(\mu)$. 
		
		We test these conditions by solving a 2SAT problem.
		We create a Boolean variable~$x_\chi$ for each
		boundary Q-node $\chi$ of either~$\skeletal_\mu(T_u)$ or~$\skeletal_\mu(T_v)$ that encodes the
		orientation of~$\chi$ as clockwise or counter-clockwise.
		For ease of notation, we also define~$x_\chi$ when~$\chi$ is
		not a Q-node but an edge.  In this case, we simply treat
		this as a placeholder for \emph{true}, i.e., both~$x_\chi$
		and~$\neg x_\chi$ are \emph{true}.
		In the following, we identify the Q-nodes of $\skeletal_\mu(T_u)$ with the
		Q-nodes of~$T_u$ they correspond to.
		We claim that the two conditions can be encoded as 2SAT
		formulas over the variables~$x_\chi$.
		
		Concerning Condition~(i), we note that we seek for an ordering $\sigma$ of the virtual edges such that
        $\sigma \in \consistent(\skeleton_\mu(T_u))$ and its
        reversal $\sigma^r$ satisfies
        $\sigma^r \in \consistent(\skeleton_\mu(T_v))$.  This can be modeled as an instances of {\sc Simultaneous PQ-Ordering}~\cite{br-spoace-16} that has two nodes~$\skeleton_\mu(T_u)$ and~$\skeleton_\mu(T_v)$ and a reversing arc~$(\skeleton_\mu(T_u),\skeleton_\mu(T_v))$ with the identity as mapping.  Then the solutions to this instance are exactly the pairs of circular orderings represented by the respective trees that are the reversal of each other.  The existence of a corresponding
        2SAT formula that describes the constraints on the orientations of the Q-nodes then follows immediately from the work of
        Bl\"asius and Rutter~\cite[Lemma 4]{br-spoace-16}, who
        refer to these formulas as $Q$-constraints.

		Concerning Condition~(ii), consider a child~$\nu_i$ and let~$\chi$
		and~$\chi'$ denote the boundaries of~$\nu_i$ in $T_u$
		and~$T_v$, respectively.  Observe that the subset of
		values~$(o_\chi,o_{\chi'}) \in \{0,1\}^2$ for which
		$\langle T_u, T_v, o_\chi,o_{\chi'} \rangle \in
		\Psi(\nu_i)$ is a subset of~$\{0,1\}^2$, and it can hence be
		encoded as the satisfying assignments of a 2SAT
		formula~$\varphi_i$ over variables~$x_\chi$ and~$x_{\chi'}$.
		Let now $\chi$ and~$\chi'$ denote the boundaries of~$\mu$
		in~$T_u$ and~$T_v$, respectively.  It follows
		that~$\langle T_u, T_v, o_u, o_v \rangle \in \Psi(\mu)$ if and
		only if there exists a satisfying assignment
		of~$\varphi_\mu \wedge \bigwedge_{i=1}^h \varphi_i$ such
		that $x_\chi = o_u$ and~$x_{\chi'} = o_v$.  These
		values of~$o_u$ and~$o_v$ can be computed by using a
		linear-time 2SAT algorithm.
		Since the number of virtual edges is $O(n)$, so is the
		number of Q-nodes, and therefore the number of variables.
		Therefore, for each pair of trees $(T_u,T_v)$ the 2SAT
		formula can be constructed and solved in $O(n)$ time.  This implies the time complexity in the statement.
		\qed
	\end{proof}

	\begin{lemma}\label{le:Rnode}
		Let $\mu$ be an R-node with children $\nu_1, \nu_2, \dots, \nu_k$. Given $\Psi(\nu_1), \Psi(\nu_2),$ $\dots \Psi(\nu_k)$, the set $\Psi(\mu)$ can be computed in $O(D_{\mathrm{max}}^{\frac{3}{2} b} \cdot n_\mu^2 + n_\mu^3)$ time, where $D_{\mathrm{max}} = \max_{v\in V}|D(v)|$, $b$ is the branchwidth of $G_\mu$, and $n_\mu$ is the number of vertices of $G_\mu$.
	\end{lemma}

	\begin{proof}
		Since $\mu$ is an R-node, $\skeleton(\mu)$ has only two possible planar embeddings. Let $u$ and $v$ be the poles of $\mu$.
		Let~$\nu_i$ ($1 \leq i \leq k$) be a child of~$\mu$ that corresponds to a virtual
		edge $(x,y)$ of $\mathcal{T}$ and let~$T_x \in D_\mu(x)$. Recall
		that~$E_{\nu_i}^\star(x)$ is a consecutive set of leaves in~$T_x$. If~$\mathcal{B}(E_{\nu_i}^\star(x))$ in $T_x$ is a Q-node~$\chi$, by Lemma~\ref{le:boundary} there are at least two edges incident to $\chi$ that do not belong to $E_{\nu_i}^\star(x)$. It follows that an orientation $o_x$ of~$\chi$ determines an embedding
		of~$\skeleton(\mu)$. 
		We call the pair $(T_x,o_x)$ \emph{compliant} with a planar embedding~$\emb_\mu$ of $\skeleton(\mu)$ if either the
		boundary is an edge, or if the orientation of the boundary
		Q-node~$\chi$ determines the embedding~$\emb_\mu$
		of~$\skeleton(\mu)$.
		We denote by~$\Psi_{\emb_\mu}(\nu_i)$ the subset
		of tuples $\langle T_x,T_y,o_x,o_y \rangle \in \Psi(\nu_i)$ such that $T_x$ with
		orientation $o_x$ and~$T_y$ with orientation $o_y$ are both
		compliant with~$\emb_\mu$.
		Similarly $\Psi_{\emb_\mu}(\mu)$ is the subset of tuples $\langle T_u,T_v,o_u,o_v \rangle \in\Psi(\mu)$ whose pairs $(T_u,o_u)$ and $(T_v,o_v)$ are both compliant with~$\emb_\mu$. 	

	We show how to compute $\Psi_{\emb_\mu}(\mu)$ from the sets~$\Psi_{\emb_\mu}(\nu_i)$
	of the children~$\nu_i$ of~$\mu$ ($1 \le i \le k$). Set $\Psi_{\emb_\mu'}(\mu)$ is computed analogously. Note that the set $\Psi_{\emb_\mu}(\nu_i)$ can be extracted by scanning $\Psi(\nu_i)$ and selecting only those admissible tuples whose pairs $(T_x,o_x)$ and $(T_y,o_y)$ are both compliant with~$\emb_\mu$.
	Since $G_\mu$ has branchwidth $b$, $\skeleton(\mu)$ is planar, it has branchwidth at most $b$, and we can execute a sphere-cut decomposition of width at most $b$ \cite{dpbf-eeapg-10} of the planar embedding $\emb_\mu$ of $\skeleton(\mu)$. 
	Such a decomposition recursively divides $\skeleton(\mu)$ into two subgraphs, each of which is embedded inside a topological disc having at most $b$ vertices on its frontier.
	The decomposition is described by a rooted binary tree, called the \emph{sphere-cut decomposition tree} and denoted as $T_{sc}$.
	The root of $T_{sc}$ is associated with $\skeleton(\mu)$; the leaves of $T_{sc}$ are the edges of $\skeleton(\mu)$; any internal node $\beta$ of $T_{sc}$ is associated with the subgraph of $\skeleton(\mu)$ induced by the leaves of the subtree rooted at $\beta$.
	Tree $T_{sc}$ is such that when removing any of its internal edges, the two subgraphs induced by the leaves in the resulting subtrees share at most $b$ vertices. We denote as $\skeleton(\beta)$ the subgraph associated with a node $\beta$ of $T_{sc}$ and with $\mathcal{D}_\beta$ the topological disc that separates $\skeleton(\beta)$ from the rest of $\skeleton(\mu)$.
	Note that $\skeleton(\beta)$ has at most $b$ vertices on the frontier of $\mathcal{D}_\beta$. In particular, if $\beta$ is the root of $T_{sc}$, $\skeleton(\beta)$ coincides with $\skeleton(\mu)$ and the vertices of $\skeleton(\beta)$ on the frontier of $\mathcal{D}_\beta$ are exactly the poles $u$ and $v$ of $\mu$.

	We compute $\Psi_{\emb_\mu}(\mu)$ by visiting $T_{sc}$ bottom-up. We equip each node $\beta$ of $T_{sc}$ with a set of tuples $\Psi_{\emb_\mu}(\beta)$, each one consisting of at most $b$ pairs of elements $(T_x,o_x)$ such that $(T_x,o_x)$ is compliant with $\emb_\mu$, and $(T_x,o_x)$ belongs to some $\Psi_{\emb_{\mu}}(\nu_i)$.	The set of tuples associated with the root of $T_{sc}$ is therefore the set $\Psi_{\emb_{\mu}}(\mu)$. Let $\beta$ be the currently visited node of $T_{sc}$.  If $\beta$ is a leaf, it is associated with an edge representing a child $\nu_i$ of $\mu$ in $\mathcal{T}$ and $\Psi_{\emb_\mu}(\beta) = \Psi_{\emb_{\mu}}(\nu_i)$. 

If $\beta$ is an internal node of $T_{sc}$, we compute $\Psi_{\emb_\mu}(\beta)$ from the sets of tuples $\Psi_{\emb_\mu}(\beta_1)$ and $\Psi_{\emb_\mu}(\beta_2)$ associated with the two children $\beta_1$ and $\beta_2$ of $\beta$.
	\begin{figure}[tb]
		\centering
		\includegraphics[width=.37\textwidth]{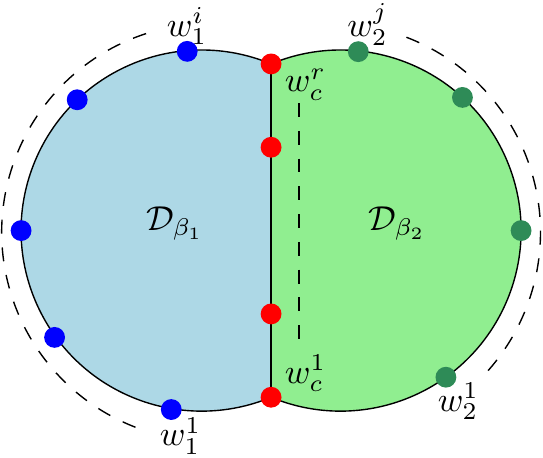}
		\caption{An example illustrating two topological discs $\mathcal{D}_{\beta_1}$ and $\mathcal{D}_{\beta_2}$ containing two subgraphs $\skeleton(\beta_1)$ and $\skeleton(\beta_2)$. $B_1=\{w_1^1,\dots,w_1^i,w_c^1,\dots,w_c^r\}$, $B_2=\{w_2^1,\dots,w_2^j,w_c^1,$ $\dots,w_c^r\}$, $B=\{w_c^1,\dots,w_c^r\}$.
		}\label{fi:disks}
	\end{figure}
		Let $B_1=\{w_1^1,\dots,w_1^i,w_c^1,\dots,w_c^r\}$ be the set of vertices of $\skeleton(\beta_1)$ that lie on the frontier of $\mathcal{D}_{\beta_1}$, and let $B_2=\{w_2^1,\dots,w_2^j,w_c^1,$ $\dots,w_c^r\}$ be the set of vertices of $\skeleton(\beta_2)$ that lie on the frontier of $\mathcal{D}_{\beta_2}$; see Figure~\ref{fi:disks} for an illustration. Let $\{w_1^1,\dots,w_1^i,w_c^1,$ $\dots,w_c^r,w_2^1,\dots,w_2^j\}$ be the set of vertices of $B_1 \cup B_2$. Also, let $B=\{w_c^1,\dots,w_c^r\}$ be the set of vertices that lie on the frontier of $\mathcal{D}_{\beta_1} \cap \mathcal{D}_{\beta_2}$; note that $B$ consists of at most $b$ vertices, i.e., $r \leq b$, and $B \subseteq B_1 \cup B_2$.
		A tuple $\langle T_{w_1^1}, \dots, T_{w_1^i}, T_{w_c^1},\dots,T_{w_c^r},o_{w_1^1}, \dots, o_{w_1^i},o_{w_c^1},\dots,o_{w_c^r} \rangle \in \Psi_{\emb_\mu}(\beta_1)$ consists of pairs $(T_{w_1^l},o_{w_1^l})$ and pairs $(T_{w_c^h},o_{w_c^h})$ ($1 \leq l \leq i$, $1 \leq h \leq r$) that are compliant with $\emb_\mu$. Similarly, a tuple $\langle T_{w_2^1}, \dots, T_{w_2^j}, T_{w_c^1},\dots,T_{w_c^r},o_{w_2^1},$ $ \dots, o_{w_2^j},o_{w_c^1},\dots,o_{w_c^r} \rangle \in \Psi_{\emb_\mu}(\beta_2)$ consists of pairs $(T_{w_2^q},o_{w_2^q})$ and pairs $(T_{w_c^h},o_{w_c^h})$ ($1 \leq q \leq j$, $1 \leq h \leq r$) that are compliant with $\emb_\mu$.
		We store the tuples of $\Psi_{\emb_\mu}(\beta_1)$ in a table $\tau_1$ where each entry is a tuple and each column contains a pair $(T_x,o_x)$.
		A table $\tau_2$ is built analogously to store the tuples $\Psi_{\emb_\mu}(\beta_2)$. We sort $\tau_1$ and $\tau_2$ according to the columns associated with the pairs $(T_{w_c^h},o_{w_c^h})$ ($1 \leq h \leq r$) and we obtain a new table $\tau$ by performing a join operation on the columns that $\tau_1$ and $\tau_2$ have in common; we then select those tuples whose pairs $(T_{w_c^h},o_{w_c^h})$ are compliant with $\emb_\mu$. Finally, we compute the set  $\Psi_{\emb_{\mu}}(\beta)$, by projecting $\tau$ on the columns associated with the pairs $(T_{w_1^l},o_{w_1^l})$ and $(T_{w_2^q},o_{w_2^q})$ ($1 \leq l \leq i$, $1 \leq q \leq j$).
		
		Observe that $\tau_1$ consists of $O(D_{\mathrm{max}}^{(i+r)})$ tuples, and table $\tau_2$ consists of $O(D_{\mathrm{max}}^{(j+r)})$ tuples.
		The join operation between $\tau_1$ and $\tau_2$ gives rise to a table $\tau$ that has  $O(D_{\mathrm{max}}^{(i+j+r)})$ tuples; since $i+r \leq b$, $j+r \leq b$, and $i+j \leq b$, we have that $2i +2j +2r \leq 3b$ and thus $i +j + r \leq \frac{3}{2} b$.

		 Sorting of the two tables can be executed in $O(D_{\mathrm{max}}^b \log(D_{\mathrm{max}}^b))$ time, since $i+r \leq b$, and $j+r \leq b$. The join operation on the sorted tables can be executed in $O(D_{\mathrm{max}}^{\frac{3}{2} b})$ time. Also, selecting those tuples of $\tau$ for which $(T_{w_c^h},o_{w_c^h})$ is compliant with $\emb_\mu$ can be done in $O(n_{\mu})$ time per tuple by looking at the cyclic order of the edges incident to $w_c^h$ in $\emb_{\mu}$ ($1 \leq h \leq r$). It follows that the set $\Psi_{\emb_{\mu}}(\beta)$ for a node $\beta$ can be computed in $O(D_{\mathrm{max}}^{\frac{3}{2} b} \cdot n_{\mu})$ time.
		 Since this procedure is repeated for every internal node of $T_{sc}$, since $T_{sc}$ has $O(n_{\mu})$ nodes, and since  $T_{sc}$ can be constructed in $O(n_{\mu}^3)$ time, we have that computing the set $\Psi_{\emb_\mu}(\mu)$ can be executed in $O(D_{\mathrm{max}}^{\frac{3}{2} b} \cdot n_{\mu}^2 + n_{\mu}^3)$ time (see, e.g.,~\cite{gt-obdpg-08,st-crr-94} for an algorithm to compute $T_{sc}$). Since  $\Psi_{\emb'_\mu}(\mu)$ is computed by an analogous procedure, the time complexity in the statement follows.\qed
	\end{proof}	

	\begin{theorem}\label{th:pqchoosable-branchwidth}
		Let $(G,D)$ be a biconnected FPQ-choosable (multi-)graph such that $G=(V,E)$ and $|V|=n$. Let $D(v)$ be the set of FPQ-trees associated with vertex $v\in V$. There exists an $O(D_{\mathrm{max}}^{\frac{3}{2} b} \cdot  n^2 + n^3)$-time algorithm to test whether $(G,D)$ is FPQ-choosable planar, where $b$ is the branchwidth of $G$ and $D_{\mathrm{max}} = \max_{v\in V}|D(v)|$.
	\end{theorem}
	\begin{proof}	
	While visiting $\mathcal{T}$, we check the existence of the admissible tuples for a node $\mu$ of $\mathcal{T}$ as shown by Lemmas~\ref{le:Snode},~\ref{le:Pnode}, or~\ref{le:Rnode}, depending on whether $\mu$ is an S-, P-, or R-node.
	Recall that for any Q-node $\mu$ that is not the root of $\mathcal{T}$ and that has poles $u$ and $v$, we have $\Psi(\mu)=D(u) \times D(v) \times \{0,1\} \times \{0,1\}$.
	It follows that the tuples that are admissible for a Q-node can be computed in $O(D_{\mathrm{max}}^2)$ time and, hence, in $O(D_{\mathrm{max}}^2 \cdot n)$ time for all Q-nodes of $\mathcal{T}$.
	The admissible tuples for all S-nodes of $\mathcal{T}$ can be computed in $O(D_{\mathrm{max}}^2 \log (D_{\mathrm{max}}) \cdot n)$ time, the admissible tuples for all  P-nodes can be computed in $O(D_{\mathrm{max}}^2 \cdot n^2)$ time, and the admissible tuples for all R-nodes can be computed in $O(D_{\mathrm{max}}^{\frac{3}{2} b} \cdot n^2 + n^3)$ time.
	Recall that the SPQR-decomposition tree of a biconnected $n$-vertex graph can be computed in $O(n)$~time~\cite{dt-olpt-96}.
	\qed
	\end{proof}
        % possibly improve running time to $O(D_max^2 \cdot n + D_max^{3/2 b} \cdot n + n^3)$

%We exploit our FPT approach together with a result of~\cite{br-npcpcep-16} to prove that clustered planarity testing can be solved in $O(d!^{\sqrt{n}}\cdot n)$ time for a flat clustered $n$-vertex graph $G$, where $d$ is the maximum out-degree of the clusters and $\sqrt{n}$ is the branchwidth of the frame of $G$.
%

We remark that our algorithmic approach cannot be extended to simply connected graphs, since it is based on the SPQR-decomposition of the input graph $G$, that expects $G$ to be biconnected.

\section{FPQ-choosable Planarity Testing and NodeTrix Planarity Testing}\label{se:nodetrix}

The study of \pqchoosable can be applied also to address other planarity testing problems that can be modeled in terms of hierarchical embedding constraints. As a proof of concept, in this section we study the interplay between \pqchoosable and NodeTrix planarity testing. 
%We start by giving some definitions regarding NodeTrix representations of flat clustered graphs and the NodeTrix planarity testing problem.

%\begin{figure}[tb]
%	\centering
%	\subfigure[]{\includegraphics[width=.3\textwidth,page=11]{intersection-link-NodeTrix}\label{fi:intro-a}}
%	\hfil	
%	\subfigure[]{\includegraphics[width=.34\textwidth,page=12]{intersection-link-NodeTrix}\label{fi:intro-c}}
%	\caption{(a) A non-planar flat clustered graph $G$. Clusters are highlighted in blue and green. (b) A planar NodeTrix representation of $G$.}
%	\label{fi:intro}
%\end{figure}

A \emph{flat clustered graph} is a graph for which subsets of its vertices are grouped into clusters and no vertex belongs to two clusters. For example, Figure~\ref{fi:intro-a} depicts a flat clustered graph with two clusters.
In a \emph{NodeTrix representation}, each cluster is represented as an adjacency matrix, while the inter-cluster edges are simple curves connecting the corresponding matrices~\cite{hfm-dhvsn-07,bbdlpp-valg-11,ddfp-cnrcg-jgaa-17,dlpt-ntptsc-19}. If no inter-cluster edges cross, the NodeTrix representation is said to be {\em planar}. For example, Figure~\ref{fi:intro-c} shows a planar NodeTrix representation of the flat clustered graph of Figure~\ref{fi:intro-a}. 

A \emph{NodeTrix graph with fixed sides} is a flat clustered graph $G$ that admits a NodeTrix representation where, for each inter-cluster edge $e$, the sides of the matrices to which $e$ is incident are specified as part of the input. If instead the sides are not specified, $G$ is a \emph{NodeTrix graph with free sides}.
If $G$ admits a planar NodeTrix representation, then we say that $G$ is \emph{NodeTrix planar}. 
%An instance of the NodeTrix planarity testing with fixed sides problem receives as input a NodeTrix graph with fixed sides $G$ and answers the question about whether $G$ is NodeTrix planar with fixed sides.
NodeTrix planarity testing is NP-complete both in the fixed sides scenario and in the free sides scenario, even when the size of the matrices is bounded by a constant~\cite{ddfp-cnrcg-jgaa-17,bdlg-ckmepd-19}.
On the positive side, it is proved in~\cite{dlpt-ntptsc-19} that one can test in polynomial time whether a flat clustered graph is NodeTrix planar with fixed sides if the size of the matrices is bounded by a constant and if the graph obtained by collapsing each cluster into a vertex has treewidth at most two. We extend this last result to graphs having bounded treewidth (provided that the size of the clusters is bounded). To this aim we model NodeTrix planarity testing with fixed sides as a problem of \pqchoosable.

Let $G$ be a NodeTrix graph with fixed sides and with clusters $C_1, \dots, C_{n_C}$. Each permutation of the vertices of $C_i$ ($1 \le i \le n_C$) corresponds to a matrix $M_i$ in some NodeTrix representation of $G$. Note that even if the side of $M_i$ to which each inter-cluster edge is incident to is fixed, it is still possible to arbitrarily permute the edges incident to a same side and to a same vertex. For example, we can permute the two edges $f$ and $g$ incident to the right side of the matrix in Figure~\ref{fi:matrix-pqtree-a}. It follows that all the possible cyclic orders of the edges incident to $M_i$ can be described by means of an FPQ-tree, that we shall call the \emph{matrix FPQ-tree} of $M_i$, denoted as $T_{M_i}$.

Namely, $T_{M_i}$ consists of an F-node $\chi_c$ connected to $4|M_i|$ P-nodes representing the vertices of $C_i$; see, e.g., Figure~\ref{fi:matrix-pqtree-b}. These P-nodes around $\chi_c$ appear in the clockwise order that is defined by $M_i$, namely $x_1^{\tau}, \dots, x_{|M_i|}^{\tau}, x_1^{\rho}, \dots, x_{|M_i|}^{\rho},$ $ x_{|M_i|}^{\beta},$ $\dots,x_1^{\beta}, x_{|M_i|}^{\lambda},\dots,x_1^{\lambda}$, where $\tau$, $\rho$, $\beta$, and $\lambda$ represent the top, right, bottom, and left side of $M_i$, respectively. Any inter-cluster edge incident to a vertex $v$ of $M_i$ corresponds to a leaf of $T_{M_i}$ adjacent to $x_v^s$ ($1\le v\le |M_i|$, $s \in \{\tau,\rho,\beta,\lambda$\}).

	\begin{figure}[tb]
		\centering
		\subfigure[]{\includegraphics[width=.32\textwidth,page=1]{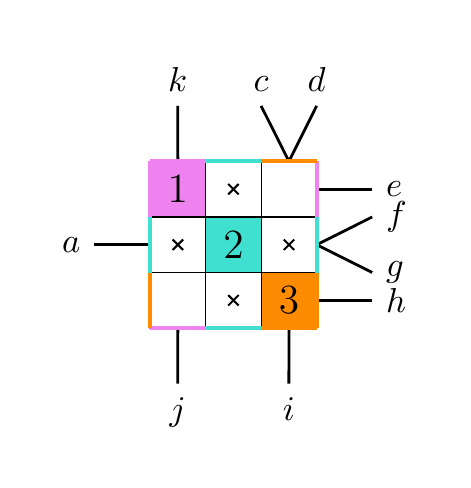}\label{fi:matrix-pqtree-a}}
		\hfill
		\subfigure[]{\includegraphics[width=.32\textwidth,page=3]{matrix-pqtree}\label{fi:matrix-pqtree-b}}
		\hfill
		\subfigure[]{\includegraphics[width=.32\textwidth,page=5]{gadgets}\label{fi:matrix-pqtree-c}}
		\caption{(a) A matrix $M_i$; (b) the matrix FPQ-tree $T_{M_i}$; (c) the gadget  $W^v$ replacing $T_{M_i}$.}
		\label{fi:matrix-pqtree}
	\end{figure}
	
\begin{figure}[tbp]
	\centering
	\subfigure[]{\includegraphics[width=.46\textwidth,page=1]{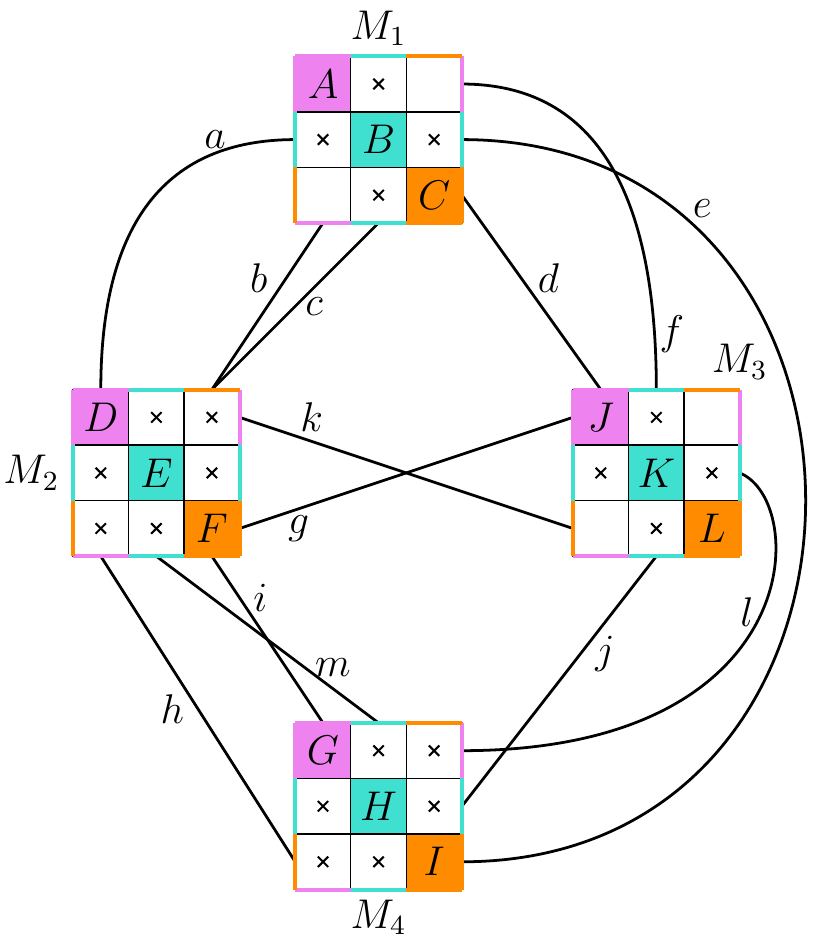}\label{fi:nodetrix-pqchoosable-a}}
	\hfil
	\subfigure[]{\includegraphics[width=.5\textwidth,page=2]{nodetrix-pqchoosable}\label{fi:nodetrix-pqchoosable-b}}
	\\
	\subfigure[]{\includegraphics[width=.3\textwidth,page=27]{nodetrix-pqchoosable}\label{fi:nodetrix-pqchoosable-c}}
	\hfil
	\subfigure[]{\includegraphics[width=.3\textwidth,page=28]{nodetrix-pqchoosable}\label{fi:nodetrix-pqchoosable-d}}
	\hfil
	\subfigure[]{\includegraphics[width=.3\textwidth,page=29]{nodetrix-pqchoosable}\label{fi:nodetrix-pqchoosable-e}}
	\hfil
	\subfigure[]{\includegraphics[width=.3\textwidth,page=30]{nodetrix-pqchoosable}\label{fi:nodetrix-pqchoosable-f}}
	\hfil
	\subfigure[]{\includegraphics[width=.3\textwidth,page=31]{nodetrix-pqchoosable}\label{fi:nodetrix-pqchoosable-g}}
	\hfil
	\subfigure[]{\includegraphics[width=.3\textwidth,page=32]{nodetrix-pqchoosable}\label{fi:nodetrix-pqchoosable-h}}
	\caption{(a) A NodeTrix graph with fixed sides $G$; (b) the constraint graph $G_C$ of $G$; (c)-(h) the FPQ-trees associated with the vertex $v_1$ of $G_C$.}
	\label{fi:nodetrix-pqchoosable}
\end{figure}

The \emph{constraint graph} of a NodeTrix graph with fixed sides $G$, denoted as $G_C$, is the FPQ-choosable multi-graph defined as follows. Graph $G_C$ has $n_C$ vertices, each one corresponding to one of the clusters of $G$, and in $G_C$ there is an edge $(u,v)$ for each inter-cluster edge that connects the two clusters corresponding to $u$ and to $v$ in $G$.
Each vertex $v$ of $G_C$ is associated with a set $D(v)$ of $|C_v|!$ FPQ-trees. More precisely, for each permutation $\pi$ of the vertices of $C_v$, let $M_v^{\pi}$ be the matrix associated with $C_v$. For each such a permutation, we equip $v$ with the matrix FPQ-tree of $M_v^{\pi}$. 

Figure~\ref{fi:nodetrix-pqchoosable-a} shows a NodeTrix graph with fixed sides $G$ whose constraint graph is depicted in Figure~\ref{fi:nodetrix-pqchoosable-b}. In Figure~\ref{fi:nodetrix-pqchoosable-b}, each vertex $v_i$ of $G_C$ ($1\le i \le 4$) represents a $3 \times 3$ matrix $M_i$ of the graph $G$ of Figure~\ref{fi:nodetrix-pqchoosable-a}; hence, $v_i$ is associated with six FPQ-trees, one for each possible permutation of the rows and the columns of $M_i$. For example, the FPQ-trees of $v_1$ are those depicted in Figure~\ref{fi:nodetrix-pqchoosable}(c)-(h).

\begin{theorem}\label{th:nodetrix-branchwitdh}
Let $G$ be a flat clustered $n$-vertex graph whose clusters have size at most $k$. Let $t$ be the treewidth of $G$. If  the constraint graph of $G$ is biconnected, there exists an $O(k!^{\frac{9}{4} t}\cdot n^2 + n^3)$-time algorithm to test whether $G$ is NodeTrix planar with fixed sides.
\end{theorem}
\begin{proof}
	Let $n_C$ be the number of vertices of $G_C$. We show that $G$ is NodeTrix planar with fixed sides if and only if $G_C$ is FPQ-choosable planar. This, together with the observation that $n_C \in O(n)$, Theorem~\ref{th:pqchoosable-branchwidth}, and the fact that if a graph has bounded branchwidth $b$ it has treewidth at most $\big \lfloor {\frac{3}{2} b} \big \rfloor -1$~\cite{rs-gm-91}, implies the statement.
	
	If $G_C$ is FPQ-choosable planar, there exists a tuple of FPQ-trees $\theta_{n_C}$ that is admissible for $G_C$. Therefore, one can associate each vertex of  $G_C$ with its FPQ-tree in $\theta_{n_C}$, execute the embedding constrained planarity testing algorithm by Gutwenger et al.~\cite{gkm-ptoei-08} and obtain a positive answer. By this technique, each FPQ-tree $T_{u}$ is replaced by a gadget $W^u$ that is built as follows. Each F-node $\chi$ is replaced with a wheel $H_\chi$ whose external cycle has a vertex for each edge incident to $\chi$. Each vertex of $H_\chi$ has an edge, called \emph{spoke}, that is incident to it and that is embedded externally to the wheel. For example, Figure~\ref{fi:matrix-pqtree-c} shows the gadget corresponding to the FPQ-tree of Figure~\ref{fi:matrix-pqtree-b}.
	Each P-node $\rho$ of $T_u$ is represented in the gadget $W^u$ as a vertex $v_\rho$ that has a spoke for each edge of $\rho$. For example, the P-node $\rho$ with incident edges $f$ and $g$ of Figure~\ref{fi:matrix-pqtree-b} is represented in Figure~\ref{fi:matrix-pqtree-c} with a vertex $v_\rho$ with two spokes $f$ and $g$.
	By performing this replacement for each FPQ-tree of $\theta_{n_C}$ and by connecting the spokes of the gadgets that correspond to the same edge, we obtain a graph $\hat{G_C}$. Gutwenger et al.~\cite{gkm-ptoei-08} show that $G_C$ is planar with the embedding constraints if and only if $\hat{G_C}$ is a planar graph.
	In order to obtain a planar NodeTrix representation, we compute a planar embedding of $\hat{G_C}$ and replace each gadget $W^u$ (corresponding to cluster $C_u$) by a matrix as follows.
	Let $W_x^u$ be a wheel of $W^u$, and let $t_1$, $t_2$, $\dots$, $t_{|C_u|}$, $r_1$, $r_2$, $\dots$, $r_{|C_u|}$, $b_{|C_u|}$, $\dots$, $b_2$, $b_1$, $l_{|C_u|}$, $\dots$, $l_2$, and $l_1$ be the spokes that are encountered by walking clockwise along the cycle of $W_x^u$. Replace $W_x^u$ with a matrix $M_u$ whose vertices are placed according to the permutation $v_1, \dots, v_{|C_u|}$. The spokes of $W^u$ that are adjacent to $t_i$ ($i=1,\dots, |C_u|$) are connected to $v_i$ on the top side of $M_u$, analogously for the spokes that are adjacent to $r_i$, $b_i$, and $l_i$, are connected to $v_i$ on the right, bottom, or left side of $M_u$, respectively.
	
	By performing this replacement for each gadget of $\hat{G_C}$, we obtain a planar NodeTrix representation $G$ of the FPQ-choosable planar graph $G_C$. It follows that, if $G_C$ is FPQ-choosable planar, $G$ is NodeTrix planar with fixed sides.
	
	We now show that if $G$ is NodeTrix planar with fixed sides, then $G_C$ is FPQ-choosable planar.
	Let $\Gamma$ be a planar NodeTrix representation of $G$. Replace each matrix $M_v$ of $\Gamma$ by a vertex $v$, and connect to it all the inter-cluster edges that are incident to $M_v$.
	We obtain a planar drawing $\Gamma'$ such that the cyclic order of the edges incident to each vertex $v$ of $\Gamma'$ reflects the cyclic order of the edges incident to matrix $M_v$ in $\Gamma$.
	Such an order corresponds to one of the $|C_v|!$ FPQ-trees associated with $v$ in $G_C$ ($|C_v|$ is the number of rows and columns of $M_v$). Therefore, $G_C$ is FPQ-choosable planar
	\qed
\end{proof}

\begin{corollary}\label{co:nodetrix-free-sides}
	Let $G$ be a flat clustered $n$-vertex graph whose clusters have size at most $k$ and whose vertices have degree at most $d$. Let $t$ be the treewidth of $G$. If  the constraint graph of $G$ is biconnected, there exists an $O((k! 4^{kd})^ {\frac{9}{4}t} \cdot n^2 + n^3)$-time algorithm to test whether $G$ is NodeTrix planar with free sides.
\end{corollary}
\begin{proof}
	The number of possible configurations in which the inter-cluster edges are incident to the matrices is $k!4^{kd}$.
	Therefore, by Theorem~\ref{th:nodetrix-branchwitdh} the statement follows.
	\qed
\end{proof}

\section{Concluding Remarks and Open Problems}\label{se:open-problems}

In this paper we have studied the problem of testing when a graph $G$ is planar subject to hierarchical embedding constraints. These constraints are given as part of the input by equipping each vertex of $G$ with a set of FPQ-trees. While the problem is NP-complete even for sets of FPQ-trees having cardinality bounded by a constant and it is W[1]-hard parameterized by tree-with, for biconnected graphs it becomes fixed-parameter tractable if parameterized by both the treewidth and by the maximum number of FPQ-trees associated with a vertex. Besides being interesting on its own right, \pqchoosable can be used to model and study other graph planarity testing problems. As a proof of concept, we have applied our results to the study of NodeTrix planarity testing of clustered graphs.

%We remark that our algorithmic approach can be of use for any planarity testing problem that can be modeled by a set of FPQ-tree associated with the vertices of the graph. For example, besides NodeTrix planarity, an immediate application of our techniques is to the problem of    
% \emph{clustered planarity testing}.  We recall that clustered planarity is a well-studied problem that asks whether a given (not necessarily flat) clustered graph admits a planar drawing where each cluster is represented inside a closed region of the plane and an edge cannot cross twice the boundary of any such a region (see, e.g.,~\cite{br-npcpcep-16,cdfpp-cpcpcg-08,ddm-pcg-02,df-ecptefcgsf-09,fce-pcg-95,gls-cpecg-06} and~\cite{cd-cp-socg05} for a survey).  Theorem~\ref{th:pqchoosable-branchwidth}, together with Theorem~2 of~\cite{br-npcpcep-16}, and the fact that the branchwidth of a planar multi-graph is $O(\sqrt{n})$ imply the following.  
%%
%\begin{corollary}
%	Let $G$ be a flat clustered graph whose constraint graph $G_C$ is biconnected. There exists an  $O(d!^{\frac{3}{2} \sqrt{n}}\cdot n^3)$-time algorithm to test whether $G$ is clustered planar, where $d$ is the maximum out-degree of the clusters and $\sqrt{n}$ is the branchwidth of $G_C$.
%\end{corollary}

We mention three open problems that in our opinion are worth future studies.
\begin{itemize}
	\item Theorem~\ref{th:pqchoosable-npcomplete} is based on a reduction that associates six FPQ-trees to each vertex of a suitable instance of \pqchoosable. 
	It would be interesting to study the complexity of \pqchoosable when every vertex is associated with less than six FPQ-trees. We recall that \pqchoosable can be solved in polynomial time if $|D_{max}| = 1$~\cite{gkm-ptoei-08}.
%	\item What is the time complexity of \pqchoosable if the number of FPQ-trees associated with each vertex is between $2$ and $5$?
	\item It would be interesting to improve the time complexity stated by Theorem~\ref{th:pqchoosable-branchwidth}.
	\item It would be interesting to extend Theorem~\ref{th:pqchoosable-branchwidth} to simply connected graphs.
%	\item Study other parameters that make \pqchoosable fixed-parameter tractable.
	\item It would be interesting to apply our approach to other problems of planarity testing related with hybrid representations of clustered graphs including, for example, intersection-link representations and $(k,p)$-planar representations (see, e.g.,~\cite{addfpr-ilrg-17,dllrt-kpprhp-walcom19}).
\end{itemize}

\bigskip
\noindent\textit{Funding.} This work was partially supported by: $(i)$ MIUR, the Italian Ministry of Education, University and Research, under grant 20174LF3T8 AHeAD: efficient Algorithms for HArnessing networked Data; $(ii)$  Dipartimento di Ingegneria dell'Universit\`a degli Studi di Perugia, under grants RICBASE2017WD and RICBA18WD: ``Algoritmi e sistemi di analisi visuale di reti complesse e di grandi dimensioni''; $(iii)$ German Science Foundation (DFG), under grant Ru~1903/3-1.

\bibliography{biblio}
\bibliographystyle{elsarticle-num}

\end{document}